\documentclass[12pt]{article}
\usepackage{amsmath}
\usepackage{graphicx}
\usepackage{enumerate}
\usepackage{natbib}
\usepackage[hyphens]{url} 

\newcommand{\blind}{1}

\RequirePackage{amsthm,amsmath,amsfonts,amssymb}
\usepackage{xcolor}
\usepackage{caption}
\usepackage{subcaption}

\usepackage[T1]{fontenc}

\theoremstyle{plain}

\newtheorem{theorem}{Theorem}[section]

\newtheorem{lemma}[theorem]{Lemma}
\theoremstyle{remark}
\newtheorem{definition}[theorem]{Definition}
\newtheorem{example}{Example}
\newtheorem{remark}{Remark}

\newcommand{\hb}[1]{\scriptsize {\bf\color{orange}[HB: #1]}\normalsize}  
\newcommand{\gskom}[1]{\scriptsize {\bf \color{cyan}[GS: #1]}\normalsize}
\newcommand{\gstext}[1]{\normalsize { \color{cyan} #1}\normalsize}

\newcommand{\Nbb}{\mathbb{N}}
\newcommand{\Rbb}{\mathbb{R}}
\newcommand{\Ebb}{\mathbb{E}}
\newcommand{\Kbb}{\mathbb{K}}

\newcommand{\Ccal}{\mathcal{C}}

\newcommand{\Ecal}{\mathcal{E}}
\newcommand{\Fcal}{\mathcal{F}}

\newcommand{\Ical}{\mathcal{I}}

\newcommand{\Pcal}{\mathcal{P}}

\newcommand{\Qcal}{\mathcal{Q}}

\newcommand{\vc}{{vc}}

\newcommand{\say}[1]{``#1''}

\begin{document}

\def\spacingset#1{\renewcommand{\baselinestretch}%
{#1}\small\normalsize} \spacingset{1}


\if1\blind
{
  \title{\bf Union-Free Generic Depth for\\Non-Standard Data}
  \author{Hannah Blocher\thanks{
    The authors gratefully acknowledge the funding and support of Hannah Blocher's doctoral studies by the Evangelisches Studienwerk Villigst e.V. and the LMU Mentoring Program.}\hspace{.2cm}\\
    Department of Statistics, Ludwig-Maximilians-University Munich\\
    and \\
    Georg Schollmeyer \\
    Department of Statistics, Ludwig-Maximilians-University Munich\\}
  \maketitle
} \fi

\if0\blind
{
  \bigskip
  \bigskip
  \bigskip
  \begin{center}
    {\LARGE\bf Union-Free Generic Depth for \\ Non-Standard Data}
\end{center}
  \medskip
} \fi

\bigskip
\begin{abstract}
Non-standard data, which fall outside classical statistical data formats, challenge state-of-the-art analysis. Examples of non-standard data include partial orders and mixed categorical-numeric-spatial data. Most statistical methods required to represent them by classical statistical spaces. However, this representation can distort their inherent structure and thus the results and interpretation. For applicants, this creates a dilemma: using standard statistical methods can risk misrepresenting the data, while preserving their true structure often lead these methods to be inapplicable. 
To address this dilemma, we introduce the union-free generic depth (ufg-depth) which is a novel framework that respects the true structure of non-standard data while enabling robust statistical analysis. The ufg-depth extends the concept of simplicial depth from normed vector spaces to a much broader range of data types, by combining formal concept analysis and data depth.
We provide a systematic analysis of the theoretical properties of the ufg-depth and demonstrate its application to mixed categorical-numerical-spatial data and hierarchical-nominal data. The ufg-depth is a unified approach that bridges the gap between preserving the data structure and applying statistical methods. With this, we provide a new perspective for non-standard data analysis.
\end{abstract}

\noindent%
\vfill

\newpage

\section{Introduction}

Modern statistical analysis frequently encounters \textit{non-standard data}, which are data that are not given in classical statistical data formats such as nominal, ordinal, interval, or ratio scales, see, e.g., \cite{stevens46}. Examples include multivariate data combining spatial and ordinal components or (partial) preference orders, where we observe a set of orders on fixed items. Addressing such data often requires either (implicitly) imposing additional assumptions, such as a metric space, see, e.g., the discussion in \cite{blocher23}, or transforming the data at the cost of losing information, such as discretizing continuous variables see, e.g., \cite{foss19, zhang24}. Moreover, these methods are generally limited to specific types of non-standard data and lack a unified framework for broader applicability, see~\cite{stumme23}.

This limitation highlights a significant research gap: The absence of a general, flexible framework that reflects the inherent structure of non-standard data while avoiding unwanted assumptions or information loss. This gap creates a fundamental dilemma in statistical analysis. On the one hand, applying standard statistical methods can distort the underlying data structure and with it the results and interpretations. On the other hand, accounting for the true structure of the data can make standard methods inapplicable. Resolving this dilemma requires a novel approach that balances the data structure with the practical requirements of statistical analysis.

This article addresses the dilemma by introducing a novel, nonparametric method -- the \textit{union-free generic depth (ufg-depth)} -- offering a new perspective on the analysis of non-standard data. The ufg-depth unifies the treatment of diverse data types without imposing further, eventually not justified assumptions. Unlike classical methods that rely on directly using data values (as in classical statistical tests like Student's $t$-test), the ufg-depth is based on natural groupings of the data elements. These groupings serve as the foundation for defining a center-outward order for the data.

The definition of the ufg-depth is based on combining \textit{formal concept analysis (FCA)} and \textit{depth functions}. Formal concept analysis provides the necessary tools to address the challenges of analyzing non-standard data by detecting and representing relationships within the data using mathematical lattice theory, see~\cite{ganter12}. It has been successfully applied in diverse fields such as knowledge discovery, see~\cite{poelmans13}, bioinformatics, see~\cite{roscoe22}, and choice theory, see~\cite{ignatov22}.
Depth functions, on the other hand, extend the notion of quantiles in $\Rbb$ to higher dimensional normed vector spaces. They provide a measure that denotes the centrality and outlyingness of data relative to a data cloud or a given probability distribution. These functions are widely used in nonparametric statistics on $\mathbb{R}^d$, see~\cite{chebana11, liu99}. The ufg-depth introduced in this article builds on these concepts and generalizes the simplicial depth developed by~\cite{liu90}. The simplicial depth defines a centrality measure as the probability that a point lies within a randomly drawn simplex (e.g., a triangle). Generalizing this idea for non-standard data involves two key challenges: redefining the concept of \say{lying in} and defining an appropriate analogue of a simplex. Formal concept analysis provides the theoretical foundation to address both challenges, enabling the development of the ufg-depth as a robust method for analyzing non-standard data.

 By combining these two concepts, this article introduces for non-standard data a robust center-outward order that considers broader spaces than solely classical statistic spaces and provides a unified framework to analyze these data. Moreover, this framework is very flexible and universal. For instance, we apply the ufg-depth on two real-word data problems -- spatial-categorical-numerical data and hierarchical-nominal data -- and analyze it generally using centrality notions derived in~\cite{blocher23b}. There the authors provide a general mapping structure for depth functions based on formal concept analysis. Moreover, they establish a systematic basis by adapting, among others, the desirable properties defined in~\cite{zuo00, serfling00, mosler22}. 


This article is organized as follows: We first provide a detailed illustration of the conceptual strategy underlying the ufg-depth and the main definitions of formal concept analysis using two concrete examples. Next, we generally define the ufg-depth and analyze its properties, drawing on the framework established in~\cite{blocher23b}. Section~\ref{sec:examples} presents concrete applications of the method to real-world non-standard data. Finally, we conclude with a discussion of the ufg-depth's contributions and limitations. Supplementary materials include detailed introduction to formal concept analysis, proofs and further side notes.

\section{Illustration of the Concepts behind the UFG-Depth}\label{sec:illustration}

In this section, we describe the idea behind the definition of ufg-depth using a snippet of concrete data examples: The \textsc{gorillas} data set, see~\cite{funwigabga12}, containing nesting sites of gorillas and data from the German General Social Survey (GGSS) concerning occupations, see~\cite{ZA5280}. 
A detailed and complete analysis of both data sets can be found in Section~\ref{sec:examples}. Moreover, we introduce the main concept of formal concept analysis needed in this article. For more details on formal concept analysis, see the supplementary or \cite{ganter12}. 

\begin{example}\label{ex:only spatial}
	The \textsc{gorillas} data provide a point pattern consisting of a spatial component (gorilla nesting sites) and categorical, numerical observations (e.g.~vegetation or elevation). For further details on the data see Section~\ref{sec:appl_mixed}. In the following, we use an excerpt of 15 observations of the gorilla nesting sites to illustrate our approach, see~Figure~\ref{fig:goriallas_excerpt}. We start by considering only the spatial component and describe the link to simplicial depth, see~\cite{liu90}.


\begin{figure}
\begin{subfigure}{0.3\linewidth}
\centering
    \includegraphics[scale=0.4]{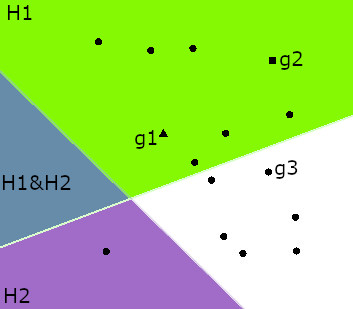} 
\end{subfigure}\hfill
\begin{subfigure}{.3\linewidth}
    \includegraphics[scale=0.4]{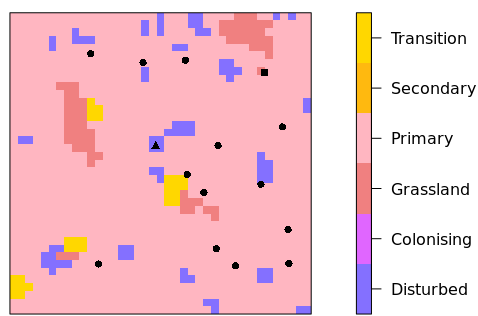} 
  \end{subfigure}\hfill
  \begin{subfigure}{.4\linewidth}
    \centering
   \includegraphics[scale=0.4 ]{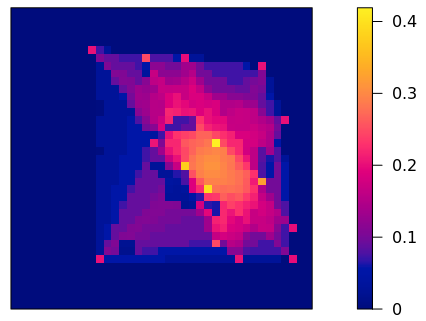}
  \end{subfigure}
  \caption{Excerpt of the spatial location of the \textsc{gorillas} data with (left) two closed half-spaces included in the plot, (middle) the vegetation component and (right) ufg-depth based on the spatial and vegetation component}
  \label{fig:goriallas_excerpt}
\end{figure}

	Formal concept analysis (FCA) serves as the foundation of our method, offering a powerful tool for uncovering relationships within data. It is based on the formalization of a cross-table. The rows of the cross-table correspond to the ground space (in formal concept analysis called \textit{objects} and denoted by $G$) and the columns represent attributes (in formal concept analysis also called \textit{attributes} and denoted by $M$) that can either be true or false for an object. The number of attributes and objects can be infinite. A cross in the cross-table indicates that the attribute holds for the element/object. These crosses are formalized by an \textit{incidence relation} $I \subseteq G \times M$. The triple $K = (G, M, I)$ is called \textit{formal context}. We want to point out that an attribute can be either true or false for an object. In particular, no degree in between can be assigned. In most cases, however, the values of the data/ground space are not binary. The transformation of many-valued data into a set of binary attributes is called \textit{(conceptual) scaling method} and has been studied extensively, for example, see~\citep[Chapter 1.3.]{ganter12}. 
	For the \textsc{Gorillas} data, the objects $G_{\Rbb^2}$ equal $\mathbb{R}^2$. Now, we have to apply a scaling method that represents the data by a set of binary attributes. Therefore, we use the method developed in~\cite{blocher23b} where the attributes $M_{\Rbb^2}$ are all topologically closed half-spaces.\footnote{In the following we always consider topologically closed half-spaces/convex sets. For simplicity, we will drop the term topological and closed from now on. Note that we will use also the term closed when referring to closed based on a closure operator.} The incidence relation $I_{\Rbb^2}$ represents whether the element lies in the half-space or not. We denote this formal context by $\Kbb_{\Rbb^2} = (G_{\Rbb^2}, M_{\Rbb^2}, I_{\Rbb^2})$. Figure~\ref{tab:gorillas_excerpt} (left) shows an excerpt of the infinite (w.r.t.~$G_{\Rbb^2}$ and $ M_{\Rbb^2}$) formal context using the indicated objects $g_1, g_2$ and $g_3$ and half-spaces $H_1$ and $H_2$ in Figure~\ref{fig:goriallas_excerpt} (left).
	
	In the next step, we use this formal context/cross-table to define the grouping procedure and resulting grouping system/groups. Let us take a subset of attributes and consider the maximal group of objects that are all valid for these attributes. This summarizes all objects that are in a certain  relation (i.e.~via the attribute subset). Thus, when taking all attribute subsets and the respective maximal object sets, we obtain a grouping system representing the relationship between the objects. In the case of the \textsc{gorilla} excerpt, we take a subset of half-spaces as a subset of attributes, and the corresponding objects are all elements that lie in every half-space of the subset. Similarly, the relationship between attributes can be considered by starting with the object set. Note that combining these two operations provides an operator that also leads to the same grouping system as starting with the attribute sets. 
	Formally, this is given by the maps $\Phi: 2^M \to 2^G, B \to B' := \{g \in G \mid \forall m \in B \colon gIm\}$ and $\Psi: 2^G \to 2^M, A \to A' := \{m \in M \mid \forall g \in A \colon gIm\}$.\footnote{For simplicity we write $\Phi(g)$ instead of $\Phi(\{g\})$ for $g \in G$. The same applies to all operators on the power set.} The grouping procedure on the object set is then given by $\gamma:= \Phi \circ \Psi$. We denote the resulting grouping system by $\Ecal = \gamma(2^G)$ and call it \textit{(set of) extents}. Now, we can exploit that the the extents $\gamma(2^G)$ and operator $\gamma$ define a \textit{closure system} and a \textit{closure operator} and have a one-to-one correspondence. A closure operator is a function on the power set $2^G$ that is extensive (i.e.~$A \subseteq \gamma(A)$ for all $A\subseteq G$), monotone (i.e.~$\gamma(B) \subseteq \gamma(A)$ for all $B \subseteq A\subseteq G$) and idempotent (i.e.~$\gamma(A) \subseteq \gamma(\gamma(A))$ for all $A\subseteq G$). A closure system is a subset $H \subseteq 2^G$ such that $G \in H$ and the intersection of elements in $H$ is again an element of $H$. 
		For $\Kbb_{\Rbb^2}$ the closure operator corresponds to the convex hull operator $\gamma_{\Rbb^2}$ that maps a set onto the smallest closed convex set containing this set. The extents $\Ecal_{\Rbb^2}$ equal all convex sets. 
	 
	 Moreover, both, the closure system and operator, can be uniquely described by a \textit{family of implications}. An implication is a statement of the form $A \to B$ with $A, B \subseteq G$ which claims that if $A$ is a subset of a group/extent, then $B$ must also be part of that group/extent, i.e.~$\gamma(B) \subseteq \gamma(A)$. Here, we call $A$ \textit{premise} and $B$ \textit{conclusion} of the implication. Reverse, when we have a set of implications $\Ical$ of a set $G$, then we say that $D \subseteq G$ \textit{respects an implication} $A \to B$ iff either $A \not\subseteq D$ or $A \subseteq D$ then $B \subseteq D$ also follows. With this, we obtain the extent set back by $\Ecal_{\Ical_G} = \{D \subseteq G \mid D \text{ respects every implication in }\Ical_G\}.$
The implications of $\Kbb_{\Rbb^2}$ are statements $A \to B$ where $B$ lies in the every convex set that contains also $A$. More formally using the convex hull operator $\gamma_{\Rbb^2}$, an implication $A \to B$ holds for the extent set given by $\Kbb_{\Rbb^2}$ iff $\gamma_{\Rbb^2}(A) \supseteq \gamma_{\Rbb^2}(B)$.
	
	However, some of these implications are redundant. For example, let $B\subsetneq \gamma_{\Rbb^2}(A)$. Then statement $A \to B$ is true but redundant since its information is already given by implication $A \to \gamma_{\Rbb^2}(A)$. These semantic redundancy structures are summarized by \cite{armstrong74} as inference axioms, see \cite[p. 45]{maier83}: Let $A,B,C,D, A_1, A_2, B_1, B_2 \subseteq G$. Then we say that the axiom of \textit{reflexivity} holds iff $A \to A$, the axiom of \textit{augmentation} holds iff $A_1 \to B$ implies $A_1 \cup A_2 \to B$, the axiom of \textit{additivity} holds iff $A \to B_1$ and $A \to B_2$ imply $A \to B_1 \cup B_2$, axiom of \textit{projectivity} holds iff $A \to B_1 \cup B_2$ implies $A \to B_1$, axiom of \textit{transitivity} holds iff $A \to B$ and $B \to C$ imply $A \to C$, and the axiom of \textit{pseodotransitivity} holds iff $A \to B$ and $B \cup C \to D$ imply $A \cup C \to D$. Note, however, that when deleting implications that follow from the above semantic structures, one may delete too many implications and end up not representing the same closure system, see Section~\ref{sec: ufg definition}. Therefore we say that a reduced family of implications $\Ical_G$ is \textit{complete} if the reduced family describes the same extent set as the unreduced one.
	
	 For the ufg-depth, we now consider the reduction of the family of implications based on reflexivity, augmentation, additivity and projectivity. With this, a complete reduction of the implications describing the spatial context is given by the set of all implications $A \to \gamma_{\Rbb^2}(A)$ with $\#A = \{2,3\}$. Based on this, we can now define the ufg-depth of $g \in \mathbb{R}^2$ as the probability that $g$ lies in $\gamma_{\Rbb^2}(C)$, where $C$ is a randomly drawn line or triangle according to an (empirical) probability measure on $\Rbb^2$. (For details on how we weight the randomly drawn triangles vs. lines, see Section~\ref{sec: ufg definition}.)
Hence, if we assume that the measure is absolutely continuous to the Lebesgue measure, we obtain the well-known simplicial depth.

\end{example}

\begin{example}\label{ex:spatial vegetation}
We consider again the \textsc{gorillas} data of Example~\ref{ex:only spatial}. Now we add the vegetation observations to the analysis. Figure~\ref{fig:goriallas_excerpt} (left and middle) show the point pattern of Example~\ref{ex:only spatial} together with its vegetation component. For clarity, the round dots represent observations in \textit{primary vegetation}, the triangle in \textit{disturbed vegetation}, and the square in \textit{grassland vegetation}. Now, the ground space is $G = \mathbb{R}^2 \times V$, where $V = \{\text{tran., sec., prim., grass., colo., dist.}\}$ consists of all possible vegetation outcomes, see Section~\ref{sec:appl_mixed} for details. 


\begin{figure}
   \begin{minipage}[b]{.45\linewidth}
    \centering
    \begin{tabular}{l | c | c | c | c | c | c}
	\: & $H_1$ & $H_2$ & $\ldots$ & $\ldots$ & $\ldots$ & $\ldots$ \\
	\hline
	$g_1$  & x & \: &\: & \: & \: \\
	$g_2$ & x & \: &\: & \: & \: \\
	$g_3$ & \: &  \: &\: & \: & x \\
	$\ldots$ & \: & \: &\: & \:  \\
\end{tabular}
  \end{minipage}\hfill
  \begin{minipage}[b]{.6\linewidth}
    \centering
    \begin{tabular}{l | c | c | c | c | c | c}
	\: & tran. & sec. & prim. & grass. & colo. & dist. \\
	\hline
	$g_1$ \: & \: & \: &\: & \: & \: & x \\
	$g_2$ \: & \: & \: &\: & x & \: & \:\\
	$g_3$ \: & \: &  \: &x & \: & \: & \: \\
	$\ldots$ \: & \: & \: &\: & \: & \: & \:\\
\end{tabular}
  \end{minipage}
  \caption{Scaling the spatial (left) and vegetation (right) component of the point pattern in Figure~\ref{fig:goriallas_excerpt} (left).}
  \label{tab:gorillas_excerpt}
\end{figure}

We proceed similarly to Example~\ref{ex:only spatial}. 
First, we use a scaling method to define a formal context/cross-table. For the spatial component we use the half-spaces discussed in Example~\ref{ex:only spatial}. The vegetation component is categorical and therefore we use the so-called \textit{nominal scaling}, see~\citep[p. 42]{ganter12} where the attributes are all possible vegetation outcomes $V$. The table in Figure~\ref{tab:gorillas_excerpt} (right) represents the vegetation part of the objects $g_1, g_2, g_3$ denoted in Figure~\ref{fig:goriallas_excerpt} (middle). Joining the two cross-tables of Figure~\ref{tab:gorillas_excerpt} by the object set $G$ gives us the formal context that represents the spatial and categorical component of each data element of $G = \mathbb{R}^2 \times V$ with attribute set $M_{\Rbb^2 \times V} = M_{\Rbb^2} \cup M_V$. The resulting groups/extents arise by considering all possible combinations of attributes and summarizing all data elements that apply to them. Thus, the extents are all sets $C \times \tilde{V}$, where $C \subseteq \mathbb{R}^2$ is a convex set and $\tilde{V} \in \binom{V}{1}\cup V$. We set $\binom{V}{1} = \{\{v_1\},\ldots, \{v_k\}\}$ for $V= \{v_1, \ldots, v_k\}$. Note that due to the nominal scaling, $\tilde{V}$ either has cardinality one or is directly the entire set. This represents the dependencies between the groups, since the relation between two vegetation categories is the same as to any other vegetation, so all other vegetation categories are also included. 

Analogously to Example~\ref{ex:only spatial}, we utilize the fact that the extents define a closure system that can be uniquely described by a family of implications. Finally, we consider only a subset of all valid implications by deleting redundancies. With this, the ufg-depth of an element $g$ is the proportion of non-redundant implications with positive empirical probability mass that imply $g$. The ufg-depth indicates how supportive/typical/central the observation $g$ is with respect to all other observations. This is because an object $g$ that lies in many non-redundant sets must have many attributes that are shared by the elements in those sets. Therefore, in the extreme case where there is an object that has all attributes, that object has maximum depth. The other extreme case, where an object has not many attributes that other objects have, denotes a small depth value. For example here, we only observed the vegetation disturbed once. Therefore, even though this observation is in the center of spatial component of the data cloud, it has a low depth value. The vegetation categories transition and grassland are never observed. So for these two categories we have that the depth is zero. The ufg-depth of an element $g \in \mathbb{R}^2 \times V$ is then given by Figure~\ref{fig:goriallas_excerpt} (right).

\end{example}

\begin{example}\label{ex: einfuehrung hierarchisch nominal}
	To highlight the wide variety of different data types that fall under the term non-standard data, we consider occupational data from the German General Social Survey (GGSS), see \cite{ZA5280}. These occupations are categorized using a hierarchy of different levels given by the International Standard Classification of Occupations (ISCO) of 2008\footnote{see \url{https://ilostat.ilo.org/methods/concepts-and-definitions/classification-occupation/} 
	 (last accessed: 14.12.2024) for details}. To define a formal context representing the different occupational groups, the occupations are successively classified into categories and subcategories. First, on a basic level (Level~1), each data element is assigned to exactly one category (coded here with digits $1,2,\ldots,9,0$, see~\citep[Appendix~D]{ZA5280}). For example, the level-1 categories of ISCO-08 are 1:~\textit{Managers}; 2:~\textit{Professionals}; etc. 
	Each of these categories is then split on a finer level (Level~2) into further subcategories and again each element of a single Level-1 category is assigned to exactly one subcategory of Level~2. For example, the Level-1 Category 3:~\textit{Technicians and associate professionals} is further divided into the Level-2 categories
	31:~\textit{Science and engineering associate professionals}; 32:~\textit{Health associate professionals}; etc. 
	Then, again the Level~2 categories are divided into further subcategories, and so on. For the ISCO-08 classification scheme, we have $4$ levels with different numbers of possible categories on each level (ranging from $1$ to $10$ categories). Now, we build a formal context for the representation of our data structure by introducing the following attributes: Every sequence $x_1x_2\ldots x_k$ with $x_i \in\{1,\ldots,9,0\}$ (and $k \in \{1, \ldots, 4\}$ for ISCO-08) describes the category on Level~$i$. For each sequence, we introduce one attribute {$x_1x_2\ldots x_k$}. An object $g$  has this attribute if it belongs to the respective occupational category up to Level~$k$. 
	Note that for this conceptual scaling, in contrast to Example~\ref{ex:only spatial}, there are usually different objects that have exactly the same attributes.
	
\end{example}

\begin{remark}
	Finally, we want to emphasize that the closure system, the implications and later the ufg-depth, depends on the application of a reasonable scaling method. The closure system and the implications provide a tool for analyzing/discussing the relational structure in detail, but the starting point is the scaling method. In particular, all the underlying assumptions of the ground space structure are determined by the scaling method. For example, in Example~\ref{ex:spatial vegetation} we have $G=\Rbb^2\times V$ as ground space. Therefore, if we have two observations in the same place with the same vegetation, we assume them to be duplicates of the same objects. Hence, the information about the two identical observations is only included in the empirical probability measure and not in the formal context itself. Another approach, not discussed here, is to consider each individual nesting point as an observation that cannot be a duplicate, but is another object in the ground space. 
\end{remark}

\section{The Union-Free Generic Depth}\label{sec: ufg + simplicial depth}
In this section, we introduce the \textit{union-free generic} depth function. It provides a centrality and outlyingness measure for data that cannot be embedded in the multidimensional real vector space. In particular, this depth function makes direct use of the relational structure provided by formal concept analysis. The definition of the union-free generic depth function is in the spirit of the simplicial depth function on $\mathbb{R}^d$, see~\cite{liu90}. We transfer the idea of using simplices to the framework of formal concept analysis.

\subsection{The Simplicial Depth from the Perspective of Formal Concept Analysis}\label{sec: ufg simplicial perspective}
Recall Example~\ref{ex:only spatial} where we discussed the formal context $\mathbb{K}_{\mathbb{R}^2} = (\mathbb{R}^2, M_{\mathbb{R}^2}, I_{\mathbb{R}^2})$ with $G = \mathbb{R}^2$ as data/objects and the set of half-spaces as attributes. In this section we have a look at the simplicial depth from the perspective of formal concept analysis.
Therefore, let $\mathcal{I}_{\mathbb{R}^2, \text{ufg}} = \{C \to \gamma_{\mathbb{R}^2}(C) \mid \begin{array}{l} C \subseteq \mathbb{R}^2 \text{ vertices of a simplex}\: \text{and} \: 2 \le \# C \le 3 \end{array} \}$ be the reduced family of implications of $\mathcal{I}_{\Rbb^2}$.
 
Using Carath{\'e}odory's theorem, see~\cite{eckhoff93}, we can first show that this family of implications describes the convex sets. Moreover, we obtain that compared to $\mathcal{I}_{\Rbb^2}$, $\mathcal{I}_{\Rbb^2, \text{ufg}}$ has deleted all implications that follow from the inference axioms of reflexivity, augmentation, additivity, and projectivity, see Lemma~2.2 in the supplementary for details. Note that the transitivity and pseodotransitivity axioms are not applied for the deletion. This is done because when restricting to the inference axioms of reflexivity, augmentation, additivity and projectivity, the information about the \textit{betweenness} of the data points is preserved directly, otherwise this information is given only indirectly. For example, let $A \subseteq \mathbb{R}^2$ with $\#A = 3$ and let $A_1, A_2 \subseteq A$ be a division of $A$ such that $A_1 \cup A_2 = A$. Then $\gamma_{\mathbb{R}^2}(A_i), i \in \{1,2\}$ is the line between the two elements of $A_i$, and using then transitivity and pseodotransitivity together, we get that $\gamma_{\mathbb{R}^2} (\gamma_{\mathbb{R}^2}(A_1) \cup \gamma_{\mathbb{R}^2}(A_2)) = \gamma_{\mathbb{R}^2}(A)$ the whole triangle.

\begin{example}\label{exampl: Rhochd}
	Note that all the considerations above for $\mathbb{R}^2$ (Example~\ref{ex:only spatial}) can be easily adopted to general $\mathbb{R}^d$ with $d \in \mathbb{N}$. The formal context is defined analogously with half-spaces in $\mathbb{R}^d$ and as closure system/operator we get again the closed convex sets with the corresponding convex closure operator. Similar to Lemma~2.2 in the supplementary we obtain all implications with $2 \le k \le d+1$ points defining a vertice of a simplex as premises of an ufg-implication.
\end{example}

With the above in mind, let us take a closer look at $$D: \mathbb{R}^d \to \mathbb{R}, g \mapsto \sum_{i=2}^{d+1} P(g \in \gamma(\{X_1\ldots X_i\}) \mid X_1\ldots X_i \text{ define vertices of a proper simplex})$$
for a probability measure $P$ on $\mathbb{R}^d$ and independent random variables $X_1, \ldots, X_{d+1}\sim P$. This gives us the sum of the probabilities that $g$ lies in a proper simplex of cardinality $2 \le k\le d+1$.
Assuming that the probability measure $P$ is absolutely continuous with respect to the Lebesgue measure, we obtain that $D$ exactly mimics the simplicial depth function. 

\subsection{Definition of the Union-Free Generic Depth}\label{sec: ufg definition}
Now we take the next step and transfer the idea based on simplicial depth to general data represented by a formal context and the resulting closure system/operator. Let $\Kbb = (G,M,I)$ be a formal context with corresponding closure system $\mathcal{E}_{\Kbb}$ and closure operator $\gamma_{\Kbb}$. In the style of the above section, we define a family of implications that is reduced based on the Armstrong rules of reflexivity, augmentation, additivity, and projectivity.
\begin{definition}\label{def: ufg family of implication}
	The \textit{union-free generic family of implications} (\textit{ufg-family of implications} for short) for a formal context $\Kbb$ on a object set $G$ is defined by
	\begin{align*}
		\Ical_{G, \text{ufg}} := \{A \to \gamma_{G}(A) \mid A \textit{ fulfills (C1) and (C2)}\}
	\end{align*}
	with the conditions on $A$: (C1) $A \subsetneq \gamma_G(A)$ and (C2) for all families $(A_j)_{j \in J}$ with $A_j \subsetneq A$ for $j \in J$ we have that $\cup_{j \in J} \gamma_G(A_i)\neq \gamma_G(A)$.
	For an implication $A \to B \in \Ical_{G, \text{ufg}}$, we say that $A\in \Ical^{prem}_{G, \text{ufg}}$ is the \textit{ufg-premise} and $B\in \Ical^{concl}_{G, \text{ufg}}$ the \textit{ufg-conclusion}.  For examples, see Section~\ref{sec: ufg simplicial perspective} and Section~\ref{sec:examples}.
\end{definition}



We call $\Ical_{G, \text{ufg}}$ generic following~\cite{bastide00} where they called an implication with minimal premise and maximal conclusion to be generic. The minimality of the premise follows from Condition (C2). 
Since we set the conclusion to $\gamma_G(A)$ the maximality is also directly given. The term union-free describes the idea behind the Condition (C2) as it covers more then only generic.\footnote{Such families are also called proper and contracted in \cite[p. 82]{ganter12}.}
Note that there can still exist non-redundant implications which follow semantically from other implications by use of transitivity and pseudotransitivity.

We want to point out that the ufg-family of implications does not necessarily result in a family of implications that describes the closure systems as it can reduce too much information. An example is the closure system $\Ecal_{\Nbb} = \left\{A \subseteq \mathbb{N} \mid \#A \text{ finite}\right\} \cup \mathbb{N}$ on $\mathbb{N}$. Here the family of all implications is given by $\Ical_{\Nbb} = \{A \to B \mid \#A = \infty, A \subseteq \mathbb{N}, B \subseteq \Nbb\} \cup \left\{A \to B \mid \#A < \infty, B \subseteq A \right\} $. So every implication does not satisfy Condition (C1) or (C2). Hence $\Ical_{\Nbb, \text{ufg}} = \emptyset$. Note that this can only happen if the underlying space is infinite, and even in infinite cases this limitation does not hold in general, as can be seen in the spatial case. However, when using the ufg-depth, this is another aspect that needs to be considered in the scaling method definition.

Now we transfer the idea of the simplicial depth and define the ufg-depth as weighted probability that an object/element lies in a randomly drawn ufg-conclusion. The definition is in line with the general mapping structure given in \cite{blocher23b}. To simplify the notation, we set $\Ical^{prem,j}_{\Kbb, \text{ufg}}$ to be the set of all ufg-premises given by the formal context $\Kbb$ of cardinality $j \in \mathbb{N}$ and define $f^j_g:2^G \to \{0,1\}, A \mapsto 1_{\gamma(A)}(g) 1_{\Ical^{prem,j}_{\Kbb, \text{ufg}}}(A)$ with $\Fcal^j = \{f^j_g \mid g \in G\}$.\footnote{For the consistency proof, we will later need the dual definitions $f^j_A:G \to \{0,1\}, g \mapsto 1_{\gamma(A)}(g) 1_{\Ical^{prem,j}_{\Kbb, \text{ufg}}}(A)$ and $\tilde{\Fcal}^j = \{f_A^j \mid A \subseteq G\}$.} Let $h^j:G\times \ldots \times G \to \{0,1\}, (g_1, \ldots, g_j) \mapsto 1_{\Ical^{prem,j}_{\Kbb, \text{ufg}}}(g_1, \ldots, g_j)$. Moreover, we define the functional U-statistics for every function $i$ $$U^{j}_{(g_1, \ldots, g_{n})}[i] = \begin{cases} {n \choose j}^{-1} \sum\limits_{1 \le i_1 < \ldots <i_j\le n} i(g_{i_1}, \ldots, g_{i_j}) , & j \le n \\ 0, & j > n \end{cases}.$$.
\begin{definition}\label{def: ufg depth}
	Let $G$ be a set. We set $\varkappa_G  \subseteq \{\mathbb{K} \mid G \text{ is object set of } \mathbb{K}\}$ to be a set of formal contexts with object set $G$ and $\Pcal_G$ to be a family of probability measures on $G$ such that for every $P \in \Pcal_G$ every extent of every $\Kbb\in\kappa_G$ is measurable. Moreover, we assume that for every $\Kbb \in \kappa_G$ there exists a one-to-one correspondence between the ufg-family of implications and the formal context. Let the weights $C_j \in \: ]0,\infty[$ be fix for all $j \in \Nbb$.\footnote{The weights $C_j$ can also be random depending on $P$, see~\cite{blocher23}.}
	
	Then the \textit{union-free generic depth (ufg-for short)} with $J_{P, \Kbb}  = \{j \subseteq \Nbb\mid P((X_1, \ldots, X_j) \in \Ical_{\Kbb, \text{ufg}})> 0\}$ with $X_1, \ldots, X_j \overset{i.i.d.}{\sim} P$ is given by
	\begin{align*}
		D: \left\{\begin{array}{l}
		G \times \kappa_G \times \Pcal_G \to \mathbb{R}^d, \\
		(g, \Kbb, P) \mapsto  \sum_{j \in J_{P,\Kbb}} \frac{C_j}{\Ebb[h^j]}\Ebb[f_g^j]
		\end{array} \right. .
	\end{align*}
	Where the expectation is based on the product measure of $P\in \Pcal_G$. The object(s) with the highest ufg-depth value is(are) called \textit{ufg-median}.
\end{definition}

\begin{remark}
First, the ufg-premises being finite is not necessary, but as the sum only includes finite ufg-premises infinite ufg-premises are not taken into account.
At a first glance the definition of $J_{P, \Kbb}$ seems to be tricky. However, it is sufficient to know an upper bound for $\max J_{P, \Kbb}$, since for an index that is not part of $J_{P, \Kbb}$, that part of the sum is zero by default in the empirical version. To obtain such an upper bound one can utilize the structure of the formal context, see, e.g.~Lemma~2.3. or Lemma~2.13. in the supplementary. Moreover, we want to point out that the definition of the ufg-family of implication in concrete settings can differ strongly in their complexity. Thus, in the definition of the formal context also the computation aspect should be taken into account. For details see Section~\ref{sec:examples}.

Second, note that the weights $C_j$ allow flexibility in the definition. For $C_j = 1$ for all $j \in J_{P, \Kbb}$ we obtain the conditional probabilities. 
\end{remark}


\begin{definition}\label{def: empirical ufg depth}
Let $G, \varkappa_G$ and $\Pcal_G$ as in Definition~\ref{def: ufg depth}. Let $P \in \Pcal_G$ and we set $J_{P, \Kbb}$ as in Definition~\ref{def: ufg depth}. Let $x_1, \ldots, x_n  \overset{i.i.d.}{\sim} P$ for $n \in\Nbb$. Again, we assume fixed weights $C_j \in \: ]0, \infty[$ for every $j \in \Nbb$. Then the \textit{empirical ufg-depth} is given by
\begin{align*}
		D^{(n)}: \left\{ \begin{array}{l}
		G \times \kappa_G  \to \mathbb{R}^d, \\
		(g, \Kbb) \mapsto \sum_{j \in J_{P, \Kbb}} \frac{C_j}{U^j_{(x_1, \ldots, x_n)}[h^j]} U^j_{(x_1, \ldots, x_n)}[f^j_g]
		\end{array} \right. .
	\end{align*}
	For a rigorous definition, from now on we set $r/0$ to zero for $r \in \Rbb$.
\end{definition}

From now on, $\varkappa_G$ and $\Pcal_G$ are two families where every extent from $\Kbb \in \varkappa_G$ is measurable for every $P \in \Pcal_G$. Moreover, we omit the index in $D, \Kbb, \Ical, M, I$ and $\gamma$ from now on if the ground space is clear.

\
\section{Structural Properties}\label{sec:structural_prop}
In the previous sections, we explored the connection between ufg-depth and simplicial depth. 
Here, we aim to determine exactly how the ufg-depth is a measure of centrality. Therefore, we build on the \textit{structural properties} given by~\cite{blocher23b}. These properties provide a systematic basis for discussing centrality and outlyingness for non-standard data represented via formal concept analysis. 
These structural properties address two aspects. First, the adaptation of existing desirable properties in $\mathbb{R}^d$, see~\cite{zuo00, serfling00, mosler22}. Some, such as quasiconcavity which relies on a notion of \say{lying in}, can be easily transferred. For others, e.g.~vanishing to infinity, this is not the case. Second, the structural properties cover the inherited centrality/outlyingness of the data structure itself. In this section, we analyze the ufg-depth function in terms of these structural properties. We use the examples above to provide the idea behind the structural properties. For overview, we underline the structural properties discussed in the theorem.

\subsection{Representation Properties}
The first two properties ensure that the depth functions preserve the structure imposed by the formal context on $G$.
Concretely, this means that representing the data $G$ by a different formal context, which results in the same closure system on $G$, should not affect the depth as long as the probability measure is preserved. In other words, if a different scaling method is used that represents the relationship between the data elements in the same way, then the relationship structure, and not the attributes and incidence relation used, should be crucial.
 The second part assumes that two objects having the same attributes, and therefore are not distinguishable from the perspective of formal concept analysis, need to have the same depth values.
\begin{theorem}
	Let $P, \tilde{P} \in \Pcal_G$ be two probability measures on $G$ and let $\Kbb, \tilde{\Kbb} \in \varkappa$ be two formal contexts on $G$. \\
	\underline{Invariance on the extents:} Assume that there exists a bijective and bimeasureable function $i: G \to G$ such that the extents are preserved (i.e.~$E$ extent w.r.t.~$\Kbb$ $\Leftrightarrow i(E)$ extent w.r.t.~$\tilde{\Kbb}$) and the probability as well (i.e.~$P(E) = \tilde{P}(i(E))$). Then $
		D_G(g, \Kbb, P) \le D_G(\tilde{g}, \Kbb, P) 
		\Leftrightarrow \tilde{D}_G(i(g), \tilde{\Kbb}, \tilde{P}) \le \tilde{D}_G(i(\tilde{g}), \tilde{\Kbb}, \tilde{P}) $
	is true for all $g, \tilde{g} \in G$.\\
	\underline{Invariance on the attributes:} Let $g_1, g_2 \in G$ with $\Psi_{\Kbb}(g_1) = \Psi_{\Kbb}(g_2)$, then $D(g_1, \Kbb, P) = D(g_2, \Kbb, P)$ holds.
\end{theorem}
\subsection{Order-Preserving Properties}\label{sec:stru_prop:order}
The order-preserving properties cover the idea of \say{maximality at the center}, \say{monotonicity relative to the deepest point}, and \say{quasiconcavity} properties in $\mathbb{R}^d$, see \cite{mosler22}. At the same time, these properties also represent the structure of the ground space, such as an inherited centrality/outlyingness structure. In contrast to $\mathbb{R}^d$ where no element has a predetermined tendency to be more central than another element, this can appear for non-standard data. For example, consider the case of two objects $g_1, g_2$ with $\Psi(g_1)\supseteq \Psi(g_2)$. Hence, $g_2$ has all attributes that $g_1$ has and therefore lies in every closure set that contains also $g_1$. $g_2$ is, in some sense, more specific than $g_1$ and therefore the depth of $g_2$ should be as least as high as the depth of $g_1$. The property that formalizes this is called \textit{isotonicity}. Note that in some cases center and outlying elements are then directly implied. When an element lies in every closure set it needs to have maximal depth. This property is called \textit{maximality} property. The reverse is called \textit{minimality} property and states that an object that lies only in the most general extent, i.e.~the entire set, needs to have minimal. Both properties follow directly from the isotonicity by Theorem~2 of \cite{blocher23b}.

\begin{theorem}
	Let $P \in \Pcal_G$ and formal context $\Kbb \in \varkappa$ with $g_1, g_2 \in G$ such that $\gamma_{\mathbb{K}}(\{g_1\}) \supseteq \gamma_{\mathbb{K}}(\{g_2\})$. Then the \underline{isotonicity property} $D(g_1, \Kbb, P) \le D(g_2, \Kbb, P)$ is true.
\end{theorem}

%

The isontonicity property can be seen as a pre-property for the stricter \textit{starshaped} and \textit{quasiconcavity/contourclosed} properties. As the name implies, the starshaped property is inspired by the \say{monotone relative to the deepest point} property in $\mathbb{R}^d$, see \cite{zuo00}. It says that if we have a center (e.g., given by the maximality property), then any element $g_2$ implied by the center $c$ and another element $g_1$ (i.e., $g_2 \in \gamma(c,g_1)$) has at least as high a depth as $g_1$. A depth function satisfies the quasiconcave property iff for every $\alpha \in \mathbb{R}$ the contour set $Cont_{D, \alpha} = \left\{g \in G \mid D(g, \Kbb, P) \ge \alpha\right\}$ defines an extent set. Since the convex sets correspond to the extents, this is a direct translation of the quasiconcavity property in~\cite{mosler22}. Recall that for the formal context $\Kbb_{\Rbb^d}$ (see Example~\ref{ex:spatial vegetation}) the ufg-depth coincides with the simplicial depth in $\mathbb{R}^d$. With \cite{serfling00} we immediately obtain that the ufg-depth is neither starshaped nor quasiconcave.

In many cases, however, one can easily define an adopted ufg-depth function that is starshaped or quasiconcave. We show in Section~\ref{sec: universality properties} that such an adaptation can lead to a quasiconcave depth function with as few ties as possible. One approach builds on order theory and we aim to obtain the smallest quasiconcave function that still lies above the original ufg-depth. This gives us $D^{qc}(\cdot, \Kbb, P): G \to \Rbb,x \mapsto \sup\left\{\alpha \in \mathbb{R} \mid Cont_{D, \alpha} \to g \right\}.$ 

\begin{theorem}\label{th:tildeDqc}
	Let $D(\cdot, \Kbb, P)$ be a depth based on formal concept analysis. Then $D^{qc}$ is quasiconcave.
\end{theorem}

Note that searching for the most similar quasiconcave function based on a loss function is another approach to get a a quasiconcave function, see the supplement. Also note that our focus here is on quasiconcave, but one can adapt these ideas to starshapedness.

\subsection{(Empirical) Sequence Properties}
The previous sections focused on how the structure of the formal context is represented in the data. In this section, we have a fixed formal context, but consider a sequence of (empirical) probability measures. These properties address issues such as duplication in a sample, how outlying objects influence the more central ones and consistency considerations.

First, let us assume we have a sample $(g_1, \ldots, g_n)$. 
In the first case, the \textit{reflecting duplication} property, we assume that there are two objects $g_i, g_{\ell}$ in the sample that cannot be distinguished by the formal context (i.e.~the same object is observed twice, or they have exactly the same attributes). Then the depth of $g_i$ should be higher when considering the entire sample compared to the sample where the duplication is deleted. For the second property, we assume that there is an object $g_i$ that is completely different from all other observed objects. This means that the only extent containing $g_i$ and any subset of the sample is directly the entire set $G$. Then this object $g_i$ should not affect the order of the remaining objects, in the sense that it does not matter whether it is in the sample or not. This property is called \textit{stability of the order}.

\begin{theorem}
	Let $\Kbb \in \varkappa$. Let $g_1, \ldots, g_n$ be a sample of $G$ with $n \in \mathbb{N}$. We denote with $P^{(n)}$ the empirical probability measure given by $g_1, \ldots, g_n$ and by $P^{(n, -\ell)}$ the empirical probability measure based on $g_1, \ldots, g_{\ell-1},  g_{\ell+1}, \ldots,  g_n$ with $\ell \in \{1, \ldots, n\}.$ \\
	\underline{Respecting duplication}: Let $i,\ell \in \{1, \ldots, n\}$ with $i \neq \ell$ and for every extent $E \in \mathcal{E}$ we have $g_{\ell} \in E$ iff $g_{i} \in E$. Moreover, assume that there exists $j \in J_{P, \Kbb}$ and ufg-premises $A_1, A_2 \in 2^{\{g_1, \ldots, g_{\ell-1},  g_{\ell+1}, \ldots,  g_n\}} \cap \Ical_{ufg}^{prem,j}$ with $g_i \in A_1$ and $g_i \not\in \gamma(A_2)$. Then, we have $D_G(g_i, \Kbb, P^{(n, -\ell)}) < D_G(g_i, \Kbb, P^{(n)}).$\\
		\underline{Stability of the order}: Assume that the only extents $E$ that contains $g_{\ell}$ for $\ell \in \{1, \ldots, n\}$ as well as any subset of $g_1, \ldots, g_{\ell-1},  g_{\ell+1}, \ldots$ is $E = G$. Then for $g, \tilde{g} \in \{g_1, \ldots, g_{\ell -1}, g_{\ell + 1}, \ldots, g_n\}$ we have $D_G(g, \Kbb, P^{(n)}) \le D_G(\tilde{g}, \Kbb, P^{(n)}) 
		\Leftrightarrow \tilde{D}_G(g, \Kbb, P	^{(n, -\ell)}) \le \tilde{D}_G(\tilde{g}, \Kbb, P	^{(n, -\ell)}). $
\end{theorem}

Finally, we discuss the consistency of the ufg-depth based on an i.i.d.~sample. Let $(P^{(n)})_{n \in \mathbb{N}}$ be a sequence of empirical probability measures based on i.i.d.~samples. We show that the ufg-depth is consistent when the set of all ufg-conclusions has a finite VC-dimension. The VC-dimension of a family of sets $\Ccal \subseteq 2^G$ is the largest number such that there exists a set $\{g_1, \ldots, g_{vc}\} \subseteq G$, $vc \in \Nbb$, with $\left\{C \cap \{g_1, \ldots, g_{vc}\} \mid C \in \mathcal{C}\right\} = 2^{\{g_1, \ldots, g_{vc}\}}$, see~\cite{dudley91}. 
In other words, the VC-dimension of $\Ccal$ denotes the largest possible set that can be still shattered be $\Ccal$. 

\begin{theorem}
Let $\Kbb \in \varkappa, P \in \Pcal$ and $J_{P, \Kbb}\subseteq \Nbb$ be the same as in Definition~4.2. of the main article. Let $X_1, \ldots, X_n\overset{i.i.d.}{\sim} P$. We assume that $\#J_{P, \Kbb}< \infty$. Moreover, we assume that for every $j \in J_{P, \Kbb}$ $\Ical^{concl, j}_{\Kbb,\text{ufg}}$ has finite VC-dimension. 
With this, we get the \underline{consistency property} with $\sup_{g \in G} |D(g, \Kbb, P^{(n)}) - D(g, \Kbb, P)| \to 0$ almost surely for $n$ to infinity. (We assume that this supremum is measurable.)
\end{theorem}

\subsection{Universality Properties}\label{sec: universality properties}

As discussed in Section~\ref{sec:stru_prop:order}, the ufg-depth $D$ is generally not quasiconcave, but one can work instead with the quasiconcave version $D^{qc}$ from Theorem~\ref{th:tildeDqc}. In this section we show that $D^{qc}$ provides a depth function that is quasiconcave and, in a sense, as flexible as possible, e.g.~having only ties that are actually needed for the quasiconcavity. This property is formalized by the \textit{universality properties} introduced in \cite{blocher23b}. 
The idea behind universality w.r.t.~a property $Q$, here quasiconcavity, is to say that a depth function (here, $D^{qc}$) is as flexible as possible if it can have the same orderings of the depth values like that of another arbitrary depth function $E$ with the same property  $Q$, if it is only equipped with an appropriate probability measure $P^*$. If $P^*$ is allowed to be chosen arbitrarily, then we speak about \textit{weak freenness}. If $P^*$ is only allowed to be chosen from a set of probability measures that are arbitrary close to each other, then we speak about \textit{strong freenness}, which is of course a stronger property than weak freenness. 
 For the mathematical details, we refer to \cite{blocher23b} and the supplementary. As it turns out, under some technical assumptions, the ufg-depth is approximately weakly free. Because of some technical subtleties in the assumptions and in the exact formulation of the statement we decided to move the corresponding theorem to the supplementary, see Theorem~2.9 and Remark~2, where also a short discussion about some cases in which the assumptions are fulfilled can be found. The following theorem now gives a concrete situation under which the ufg-depth is also strongly free. This is a main advantage compared to the generalized Tukey depth,\footnote{The generalized Tukey depth is based on \cite{schollmeyer17a,schollmeyer17b} and was firstly formally introduced in \cite{blocher22}. A description of the used basic concepts and an in-depth analysis of the properties of the generalized Tukey depth can be found in \cite{blocher23b}.
 	The generalized Turkey's depth $T$ of an object $g$ w.r.t.~a formal context $\mathbb{K}$ and w.r.t.~a probability measure $P$ is defined as $T(g):= 1- \sup \{P(E) \mid E \mbox{ extent of }\mathbb{K}: g \notin E\} $.} which is not strongly free, as shown in \cite[Theorem~10]{blocher23b}.  
\begin{theorem}
	Let $\Kbb = (G, M,I)$ be a formal context given by hierarchical-nominal data with $L\ge2$ levels, $K\geq 3$ categories on each level, and the scaling method presented in Example~3 of Section~2 of the main article.  We assume that for each object $g\in G$ there exists another object $\tilde{g} \in G$ with $g \neq \tilde{g}$ and $\Psi(\{g\}) = \Psi(\{\tilde{g}\})$.
	
	 We set $C_1,C_2>0$ in the ufg-depth. Then the quasiconcave version \underline{$D^{qc}$ of the ufg-depth} \underline{is strongly free with respect to the property quasiconcavity}. This means that for every $\varepsilon >0$ there exists a family $\mathcal{P}^\varepsilon $ of probability measures with diameter
	 less than or equal to $\varepsilon$ such that for any other arbitrary quasiconcave depth function $E$ and any arbitrary probability measure $P$ there exists a measure $P^* \in \mathcal{P}^\varepsilon$ such that
	$$\forall g,\tilde{g} \in G: E(g,\mathbb{K},P) >E(\tilde{g},\mathbb{K},P) \Longrightarrow D^{qc}(g,\mathbb{K},P^*) >D^{qc}(\tilde{g},\mathbb{K},P^*).$$
\end{theorem}

\section{Examples}\label{sec:examples}

In this section we provide two application examples of the ufg-depth. First, we analyze the \textsc{Gorillas} data, see Example~\ref{ex:spatial vegetation} and second, the occupational data from the German General Social Survey (GGSS), see Example~\ref{ex: einfuehrung hierarchisch nominal}.\footnote{Both analysis can be found on  GitHub: \url{https://anonymous.4open.science/r/ufg_depth_application-0567/}. (last accessed: 14.12.2024)} We want to emphasize that analyzing the ufg-depth for a concrete data type of interest can lead to a simplified definition of the ufg-depth and, in particular, improve the computation time by exploiting the further data structure. This has been done in~\cite{blocher23}, where the authors defined, analyzed and applied the ufg-depth on the special case of partial orders as ground space. 

%
%

%
\subsection{Mixed Categorical, Numeric and Spatial Data}\label{sec:appl_mixed}
Recall Example~\ref{ex:spatial vegetation}, where we used a snippet of the gorilla nesting sites data to motivate our approach. Now, we want to extend this example by adding a further covariate (elevation) and considering a larger sample. Besides, we outline how the situation of the ground space in Section~\ref{ex:spatial vegetation} was simplified.

The data set is stored in the R--package \textsc{Gorillas} and \textsc{Gorillas.Extra} and both are provided by the R-packages \textsc{spatstat}, see~\cite{baddeley05}. The observations are a point pattern where each point represents one nesting site of the gorilla population at the Kagwene Gorilla Sanctuary in Cameroon. The nesting sites where observed from 2007 and 2009 and it consists of 647 observations. For more details we refer to \cite{funwigabga12}. 
In the following, we analyze the sample of the gorilla nesting sites observed in 2006. In total we have 121 points which are plotted in Figure~\ref{fig:ppp_2006_veg} (left). Besides the spatial observation, we include the vegetation and elevation component, see Figure~\ref{fig:ppp_2006_veg} (right) and Figure~\ref{fig:ppp_elev_depth} (left). Mainly \textit{primary} (76 points) and \textit{disturbed} (24 points) vegetation category are associated to the observed points. The corresponding elevation values range from 1340 to 2053. 

\begin{figure}
  \begin{minipage}[b]{.45\linewidth}
    \centering
   \includegraphics[scale=0.4]{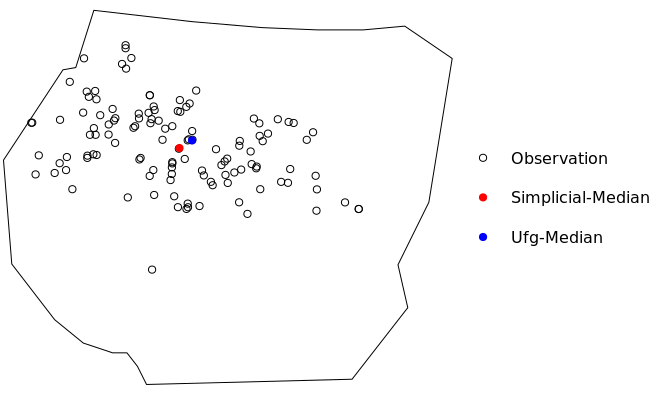} 
  \end{minipage}\hfill
  \begin{minipage}[b]{.45\linewidth}
    \centering
   \includegraphics[scale=0.7]{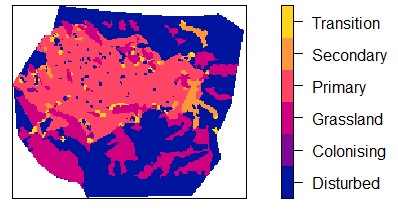} 
  \end{minipage}
   \caption{The gorilla nesting sites (left) and vegetation (right) of the Kagwene National Park (Cameroon) for the year 2006.}
   \label{fig:ppp_2006_veg}
\end{figure}	


In the illustration in Example~\ref{ex:spatial vegetation}, we simplified the data situation for the sake of accessibility. There, the underlying ground space was assumed to be $\mathbb{R}^d \times V$ with $V= \{transition, secondary, primary, grassland, colonising, disturbed\}.$ However, since we are only interested in the nesting sites of the gorillas within the Kagwene Gorilla Sanctuary in Cameroon, we now reduce the spatial set to $K \subseteq \mathbb{R}^2$ which represents the area of the national park. Moreover, at the same location, the ground space in Section~\ref{sec:illustration} assumed that two different vegetation categories are possible. This assumption does not hold as the vegetation is a fixed covariate and unique to the location part. Hence, the ground space should be $\left\{(x,v) \mid x \in K \text{ with unique corresponding vegetation } v \in V\right\} \subsetneq \Rbb^2 \times V$. Finally, we extent our analysis and add the elevation as a further observation value. With this, we get as underlying ground space
\begin{align*}
	G = \left\{(x,v,e) \big| \begin{array}{l}
	 x \in K, x \text{ with unique corresponding vegetation } v \in V \text{ and elevation } e \in \mathbb{R}
\end{array}	\right\}.
\end{align*}
The formal context now results from scaling the spatial component as in Example~\ref{ex:only spatial}, the categorical variable using nominal scaling, see Example~\ref{ex:spatial vegetation}, and for the ordinal component we use \textit{interordinal scaling}, see~\citep[p. 42]{ganter12}. The interordinal scaling equals the spatial scaling and since half-spaces in $\mathbb{R}^1$ are one-side unbounded intervals, we get as attributes \say{$\le x$} and \say{$\ge x$} for all $x \in \mathbb{R}$. Despite the changes to the ground space, the extents and implications are similar to those described in Example~\ref{ex:spatial vegetation}. Using the notation introduced in Section~\ref{sec:illustration}, we obtain as the extent set
\begin{align}\label{def:extents gorilla all covariates}
	\left\{C \times \tilde{V} \times [a,b] \mid \begin{array}{l}  C\subseteq \mathbb{R}^2 \text{ closed convex set, }
	 \tilde{V} \in \binom{V}{1}\cup V, a \le b	
\end{array}	\right\}.
\end{align} 
The set of implications are all statements $A \to B$ with $A \subseteq G$ and $ B \subseteq \gamma_{\mathbb{R}^2}\circ \pi_{\mathbb{R}^2}(A) \times \tilde{V}\times [\min\{\pi_{\mathbb{R}}(A)\},\max\{\pi_{\mathbb{R}}(A)\}]$ with $\tilde{V}= \pi_V(A)$ if $\# \pi_V(A) = 1$ and $\tilde{V} = V$ else.
So if an implication $A \to B$ is true. Then all elements in $B$ must be inside the smallest convex hull containing the spatial part of $A$. Furthermore, all elements must lie between the minimum and maximum value of the elevation component in $A$, and finally, if $A$ consists of only one vegetation class, then all elements in $B$ are of the same category.

The next step is to consider the calculation of ufg-implications $\Ical_{\text{ufg}}$. Therefore, we first utilize that the formal context can be divided into three formal contexts: the spatial, the elevation and the vegetation part. 


\begin{lemma}
	For the formal context $\Kbb_G$ with extent set given by (\ref{def:extents gorilla all covariates}), we have for the ufg-family of implications 
	\begin{align*}
		\Ical_{ufg} \subseteq \left\{A \to B \biggl| \begin{array}{l} A \subseteq \Rbb^2 \times V \times \Rbb \text{ and } 2 \le \# A \le 4, \\
		\pi_{\Rbb}(B) = [\min\{\pi_{\mathbb{R}}(A)\},\max\{\pi_{\mathbb{R}}(A)\}], \:  \pi_{\Rbb^2}(B) = \gamma_{\Rbb^2}\circ\pi_{\Rbb}(A), \\
		\pi_{V}(B) \in \binom{V}{1}\cup V:  \pi_{V}(B) =  \pi_{V}(A) \text{ if } \#  \pi_{V}(A) = 1, \pi_{V}(B) = V \text{ else}
		 \end{array} \right\}.
	\end{align*}
\end{lemma}

\begin{figure}
  \begin{minipage}[b]{.45\linewidth}
    \centering
   \includegraphics[scale=0.7]{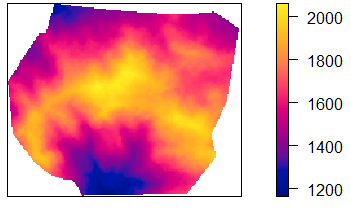} 
  \end{minipage}\hfill
  \begin{minipage}[b]{.45\linewidth}
    \centering
   \includegraphics[scale=0.85]{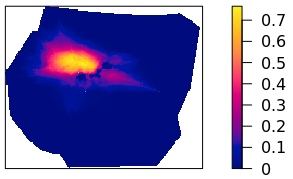} 
  \end{minipage}
   \caption{The elevation (left) of the Kagwene National Park (Cameroon) and the computed ufg-depth (right).}
   \label{fig:ppp_elev_depth}
\end{figure}



Figure~\ref{fig:ppp_elev_depth} (right) shows the calculated ufg depth for the observed point pattern in 2006. The ufg-median is unique and has a depth of $0.765$.\footnote{All values are rounded to three decimal places.} It is an observed point and lies within the primary vegetation category at an elevation of 1805. The ufg-median is thus in the most frequently observed vegetation category and is also relatively close to the median of the numerical elevation component ($52.1\%$ observed elevation values are strictly below 1805). The spatial component is also relatively close to the median of the median computed by the simplicial depth, see Figure~\ref{fig:ppp_2006_veg} (left), where only the spatial part is considered. Note that the observation with the highest simplicial depth has an elevation value of 1900 and is therefore further away from the center from the perspective of the elevation component.
The minimum depth value is zero. In particular, the ufg-depth is always zero when the elevation corresponding to a location is strictly below (or above) the minimum (or maximum) of the observed elevation values. This is the reason why the area in the center of the image, from a purely spatial perspective, has a low or even zero ufg-depth value. It can also be seen that the ufg-depth reflects that the vegetation categories colonization (1 point), grassland (7 points), secondary (5 points) and transitional (6 points) are not often observed. 
\subsection{Hierarchical-Nominal Data}\label{sec:appl_hierarchical}

As a further example, we analyze the ufg-depth for occupational data as described in Example~\ref{ex: einfuehrung hierarchisch nominal}. We use data from the German General Social Survey (GGSS) of the year 2021, see \cite{ZA5280}. Additionally, we also compare the ufg-depth to three other \textit{measures of central tendency}. Namely one approach that only uses the categories on the finest level, secondly, a \textit{top down} approach that analyses the frequencies of occupations successively, going from coarser levels to finer levels, and thirdly, the median based on the generalized Tukey depth. This comparison aims to illustrate the fact that the ufg-depth approach is different from these other approaches in a substantial way which is to some extent surprising given the meager structure of hierarchical-nominally scaled data.\footnote{For example, for hierarchical-nominal data, the formal extents are either nested or they have an empty intersection.}

For the specification of the hierarchical categories of occupation we use the International Standard Classification of Occupations (ISCO) 2008. It consists of occupational categories on $4$ levels with up to $10$ categories on each level, coded by digits $0-9$. We analyze the set of all $2700$ respondents for which the ISCO-08 status is a available. The sample is not drawn i.i.d., the respondents in east Germany were over-sampled. We account for this by simply reweighting the obtained empirical measure accordingly. Figure~\ref{fig:hierarchical} (left and right) depicts the distribution of the occupations by drawing histograms on all $4$ levels of the hierarchy of occupations.\footnote{Level~1: black; Level~2: blue; Level~3: yellow; Level~4: pink. Left: Level-1 Category 3 to Level-1 Category 4. Right: Zoom into Level-1 Category 3 (from Level-2 Category~$32$ to Level-2 Category~$33$). The height of the bars corresponds to the absolute observed frequencies (counts) within the corresponding categories on a log scale.} 
While the left figure goes from Level-1 Category~3 to Level-1 Category~4, the right picture zooms into Level-1 Category~3 and goes from Level-2 Category~$32$ to Level-2 Category~$33$. The vertical lines indicate different further measures of location (orange: occupation with the highest ufg-depth; green: category with highest frequency on the finest level; purple: median according to the top down approach (see below)). 

\begin{figure}
	\begin{minipage}[b]{.99\linewidth}
		\includegraphics[width=0.49\textwidth]{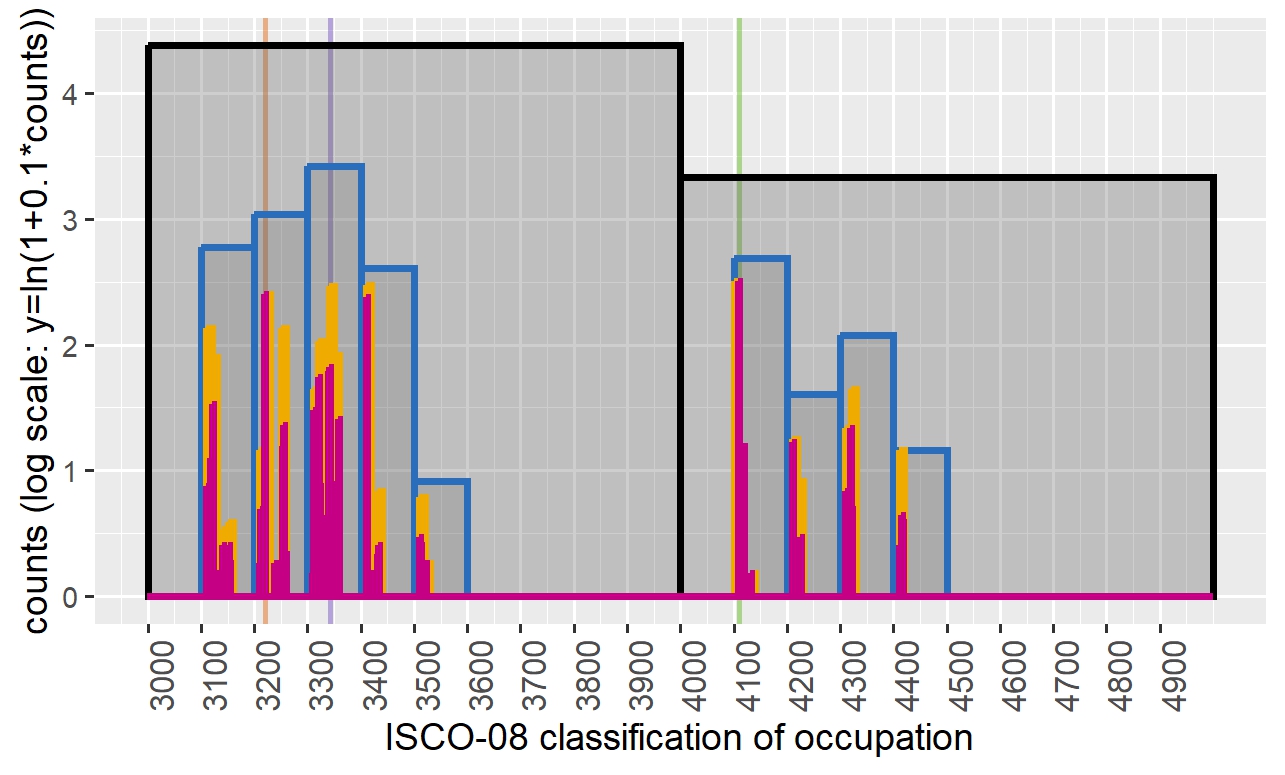} 
		\includegraphics[width=0.49\textwidth]{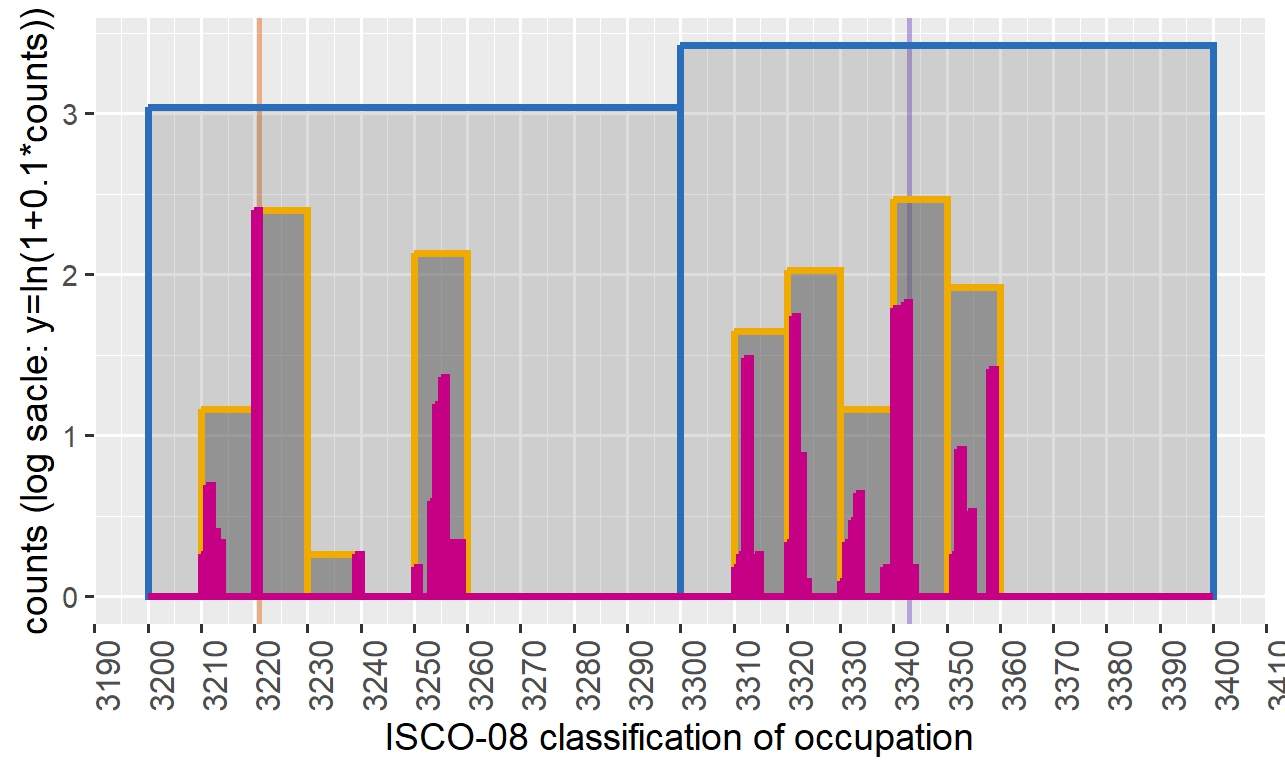} 
	\end{minipage}
	\caption{Histogram of the occupations on all 4 levels.} 
	\label{fig:hierarchical}
\end{figure}

The \textit{ufg-median}, i.e., the occupation with the highest ufg-depth, is occupation 3221: \textit{Nursing Associate Professionals} (indicated with the orange vertical line in the plots). The depth for these persons was $0.927$. The smallest ufg-depth value has occupation 6210: \textit{Forestry and related workers} with a depth value of $0.824$. 
The ufg-depth has all-together $285$ unique depth values which induce $285$ contour sets. These contour sets $Cont_{D, \alpha}$ induce three different attribute sets $\Psi(Cont_{D, \alpha})$ containing all attributes each object in the contour set has, namely these sets are $B_1 = $\{Level-4 Category 3221: \textit{Nursing associate professionals}; plus all corresponding categories on the coarser Levels 1-3\}; $B_2 = $\{Level-1 Category 3: \textit{Technicians and associate professionals}\} and the $B_3 = \emptyset$. These three attribute sets induce the corresponding extents $\Phi(\Psi(Cont_{D,\alpha})) = \gamma(Con_{D,\alpha})$ that are (due to construction) exactly the contour sets of $D^{qc}$. Note that generally, for hierarchical-nominal data with $L$ levels, a quasiconcave depth function, and in particular $D^{qc}$, can only have up to $L+1$ different depth values (c.f.~the proof of Theorem~2.11 in the supplementary). 
Here, with $3$ levels, the quasiconcave version $D^{qc}$ is more flexible compared to e.g., the generalized Tukey depth in this data situation.\footnote{The generalized Tukey depth function has only two different depth values, which is typical, compare the discussion and the proof of the non-freenness of the generalized Tukey depth in \cite{blocher23}.} 

We now compare the ufg-depth with other measures of central tendency for occupational data. First, note that the ufg-median here differs from the modus, i.e.~the occupation (at the finest level) with the highest frequency. The modus is occupation 4110:~\textit{General office clerks} (indicated by the green vertical line). The modus only considers the categories at the most detailed level (Level~4). The structure at the more general Levels 1-3 is not taken into account at all. Opposed to this, the ufg-depth does take the other levels into account: First, note that the ufg-premises are exactly the one-element sets and the sets of two objects with different occupational ISCO-08 categories, see Theorem 3.10.~in the supplementary. If we only consider the one-element ufg-premises, we end up with the frequencies at the finest levels. However, the ufg-depth approach also uses the two-element ufg-premises. Let $p=\{x,y\}$ with $x=x_1x_2\ldots x_{\ell-1} x_{\ell}\ldots x_k$ and $y= x_1 x_2\ldots x_{\ell-1} y_{\ell} \ldots y_k$ be a two-element ufg-premise, where the first occurring classification difference is $x_{\ell}$ and $y_{\ell}$. Then $p$ contributes to the depth values of each object in $\gamma(p)$, which consists of all objects that share the first $\ell -1$ category assignments with $x$ and $y$, i.e.~are also categorized in the group $x_1x_2\ldots x_{\ell-1}$. For each object that can be distinguished from $x$ and $y$ based on the first $\ell -1$ categories, the ufg-premise $p$ does not contribute to the ufg-depth of this object. Therefore, generally, the two-element ufg-premises contribute to the ufg-depth on all levels of the hierarchy. Let us now compare with a further construction of a median, which could be called the \textit{top down} approach. A simple way of ordering the hierarchical categories is to look first at Level~$1$ and take the modal category at that level, i.e.~the category with the highest frequency of occurrence (here Level~1 Category~3:~\textit{Technical and associate professionals}). Then, within this modal category, one could look at the subcategories at Level~2 and again take the modal subcategory at level $2$ (here, Level~2 Category~33:~\textit{Business and administration associate professionals}), and so on. In our data set, this approach gives us the median occupation 3343:~\textit{Administrative and executive secretaries} (indicated by the purple vertical line). Like the ufg-depth approach, the top down approach uses all levels of the hierarchical structure. However, the mode at Level~1 predetermines the Level-1 category of the final median, and unlike the ufg-approach, if a data point does not fall into the modal Level~1 category, it can never become the median, even if -- due to high frequencies -- it is a clear median candidate from the perspective of all other levels. This property of predetermination of the Level~1 mode is also shared by the median according to the \textit{generalized Tukey depth} (c.f., the proof of the failure of strong freenness of the generalized Tukey depth given in \cite{blocher23b}, Theorem~10). In addition, the generalized Tukey depth has only two different depth values for this data set. Specifically, all occupations with Level-1 Category~$3$ have a generalized Tukey depth of $ 0.747$, and all other occupations have a generalized Tukey depth of $0.0.710$.






\section{Conclusion}
Providing statistical methods that take into account the underlying data structure is essential in statistics. The ufg-depth introduced here is a non-parametric and user-friendly method that uses the theory of formal concept analysis and data depth to define a statistical method for non-standard data. While this article presented and analyzed the ufg-depth and provided two descriptive examples showing the benefits of the ufg-depth, it also raised further research questions:\\
\textbf{Statistical Inference:} With the exception of the consistency property, the analysis of the ufg-depth and the examples focus on descriptive analysis. However, building on the consistency property, a further research question is how to define statistical inference tests. These tests can build on approaches provided by ~\cite{li04} in $\Rbb^d$. \\
\textbf{Deeper analysis of the quasiconcave version of a depth function:} In Section~\ref{sec:stru_prop:order} we briefly touched on the topic of the quasiconcave version of a depth function and in Section~\ref{sec:appl_hierarchical} we showed that this does indeed provide a meaningful and non-trivial depth function. It is of interest to explore this topic in more detail, in particular with a closer look at $\Rbb^d$ and the large variety of depth functions already defined for $\Rbb^d$. \\
\textbf{Other data sets:} We applied the ufg-depth to two data types, the categorical-numerical-spatial data and the hierarchical-nominal data. In~\cite{blocher23} the authors applied the ufg-depth to partial orders. These three data types are by no means all possible non-standard data. The investigation of further data types, and in particular of scaling methods that transform the data into a formal context, is a further interesting research area.\\
\textbf{Further generalizations of depth functions in $\Rbb^d$:} So far the Tukey depth, see~\cite{tukey75}, and the simplicial depth, see~\cite{liu90} are generalized to non-standard data. Also a discussion on the convex-hull-peeling depth, see~\cite{blocher22}, has been started. Similarly, this can be done with many other depth functions, such as the projection depth, see, e.g.,~\cite{serfling00}.

\newpage

\bigskip
\begin{center}
{\large\bf SUPPLEMENTARY MATERIAL }
\end{center}

In the following, we provide the supplementary material and information to the main article \textit{Union-Free Generic Depth for Non-Standard Data}. This includes a short introduction to formal concept analysis, a further discussion on the quasiconcave ufg-depth and all the proofs of the claims made in the main article. Unless otherwise stated, all references to equations, lemmas, etc. are to the supplementary material. 

The repository corresponding to the main article can be found at \url{https://anonymous.4open.science/r/ufg_depth_application-0567/}(last accessed: 14.12.2024). There we also provide all the information about the reproducibility of the results in Section~6 of the main article.

\section{Formal Concept Analysis}
Formal concept analysis can be seen as applied lattice theory, which describes the relationship between data elements in a user-friendly and unified way. It is based on the formalization of a cross-table, see \cite[p. 17]{ganter12}:
\begin{definition}
	The triple $\mathbb{K} = (G,M,I)$ defines a \textit{formal context} with $G$ s set of \textit{objects} and $M$ a set of \textit{attributes}. $I \subseteq G \times M$ states a binary relation between $G$ and $M$.
\end{definition}

In our case, the objects $G$ correspond to the data described by the attributes $M$. Note that an object/data element can either have this attribute or not. While in some cases binary attributes, such as yes or no responses to a yes-no question, are naturally given, this is generally not the case. Therefore, we use so-called \textit{scaling methods}, see~\citep[Chapter 1.3.]{ganter12}. These methods convert non-binary information about the data into attributes with a binary incidence relation. Examples can be found in \cite{ganter12, blocher23b} and in the examples below. With the scaling method, we achieve that all types of data are presented through a formal context in a unified way. 

\begin{example}\label{exampl: spatial_context}
Recall Example~2 in the main article with $G = \mathbb{R}^2 \times V$ as ground space and attributes $M_{\mathbb{R}^2 \times V} = M_{\mathbb{R}^2} \cup M_V$. $I_{\mathbb{R}^2 \times V}$ now describes the incidence of both the spatial and the categorical component, where we say that the categorical attribute holds if the data element has that category. A snippet of this formal context is the joint (by the objects) tables of Figure~2 in the main article.
\end{example}

Especially this formalization of a cross-table is the basis to rigorously define the grouping procedure. Therefore, consider the following \textit{derivation operators}, see~\citep[p. 18]{ganter12}: 
\begin{align*}
&\Psi: 2^G \to 2^M, A \to A' := \{m \in M \mid \forall g \in A \colon gIm\} \qquad \text{and}\\ 
&\Phi: 2^M \to 2^G, B \to B' := \{g \in G \mid \forall m \in B \colon gIm\}.
\end{align*}
$\Psi$ maps a set of objects $A$ to each attribute that each object in $A$ has. $\Phi$ does the same, only with the roles of attribute set and object set reversed.
In particular, the composition $\gamma_G := \Phi \circ \Psi$ now groups the objects based on the attributes in a maximal way. More precisely, the set $\gamma_G(A)$ composes all objects that share the same attributes given by $\Psi(A)$. 
\begin{definition}
We call $\gamma_G(A)$ with $A \subseteq G$ an extent and $\Psi(A)$ an intent of $\mathbb{K}$.  Additionally, we denote the set of all extents by $\mathcal{E}_G$.
\end{definition}

The set of extents can be partially ordered using the subset relation. 
If $\gamma_G(A) \subseteq \gamma_G(B)$ for $A, B \subseteq G$, we can conclude that the objects in $\gamma_G(A)$ are more specific than those in $\gamma_G(B)$.  This means that the attributes common to all objects in $B$ are a subset of the attributes common to all objects in $A$, see \citep[Chapter 1.]{ganter12} for details. By examining this order on the extents, we can determine whether the relationship between the elements/objects is reasonable or not. We also get an idea of how fine the grouping is, i.e. if we are close to the power set. 

In addition, the set of extents defines a closure system on $G$ with the corresponding closure operator $\gamma_G$, which builds the bridge to lattice theory, see \citep[Chapter 0]{ganter12}.

\begin{definition}
	Let $G$ be a set. Then $\gamma_G: 2^G \to 2^G$ is a \textit{closure operator} on $G$ if and only if $\gamma_G$ is \textit{extensive} (for all $A \subseteq G\colon A \subseteq \gamma_G(A)$), \textit{monotone} (for all $A \subseteq B \subseteq G$ we have $\gamma_G(A) \subseteq \gamma_G(B)$) and \textit{idempotent} (for all $A\subseteq G, \: \gamma_G(A) = \gamma_G(\gamma_G(A))$).
	
	$\gamma_G(2^G)$ induces the corresponding \textit{closure system}. Closure systems are families of sets which are closed under arbitrary intersections (let $(A_j)_{j \in J} \subseteq \gamma_G(2^G)$ then $\cap_{j \in J} A_j \in \gamma_G(2^G)$) and contain the entire set $G \in \gamma_G(2^G)$.\footnote{In the following, we use both the terms \say{extent set} and \say{closure system}, depending on whether we want to emphasize that it is based on a formal context or that we exploit the mathematical structure.}
\end{definition}

Note that there exists a one-to-one correspondence between closure operators and closure systems/extents $\mathcal{E}_G$, see \citep[p. 8]{ganter12}.

\begin{example}\label{exampl: spatial_closure}
	Consider the formal context $\mathbb{K}_{\Rbb^2 \times V}$ defined in Example~\ref{exampl: spatial_context}. Then we get as set of extents $\mathcal{E}_{\Rbb^2 \times V} = \left\{C \times \tilde{V} \mid \begin{array}{l} C \text{ topologically closed convex set}\: \land \:  \tilde{V} \in \binom{V}{1} \cup V \end{array}\right\}$. Note that due to nominal scaling, $\tilde{V}$ either has cardinality one or is directly the entire set. This follows from the fact that if two different categories are grouped together, then the relation between these two categories is the same as to any other category, so all other categories are also included in order not to state a relation between these two categories that does not exist.
\end{example}

As we saw in Example~\ref{exampl: spatial_closure}, the closure system/extent set contains the structure of the data and describes the dependencies between data elements. This becomes even clearer when we exploit the fact that every closure system can be described by a family of implications. In the context of the closure operator $\gamma_G$, we define implications as follows, see~\citep[Chapter 2.3]{ganter12}:\footnote{Note that in \cite{ganter12} the authors discuss attribute implications. The results can be applied to object implications discussed here.}

\begin{definition}
Let $G$ be a set. An \textit{implication} is a tuple $(A_1, A_2) \in G \times G$. We say that $A_1$ implies $A_2$ and denote this by $A_1 \to A_2$. We call $A_1$ the \textit{premise} and $A_2$ the \textit{conclusion} of the implication $A_1 \to A_2$. 

Let $\gamma_G$ be a closure operator on $G$ with a corresponding closure system $\mathcal{E}_G$. Then, the closure system defines a family of implications consisting of statements $A_1 \to A_2$ with $\gamma_G(A_1) \supseteq \gamma_G(A_2)$. The family of all implications provided by $\mathcal{E}_G$ is denoted by $\Ical_G$.
For a given closure system $\Ecal_G$ we say that an \textit{implication $A_1 \to A_2$ holds} if and only if $\gamma_G(A_1) \supseteq \gamma_G(A_2)$. 
\end{definition}

 
\begin{example}\label{exampl: spatial_implication}
	Recall Example~\ref{exampl: spatial_context} and~\ref{exampl: spatial_closure}.
	For $\Kbb_{\Rbb^2 \times V}$ we have as family of all implications
	\begin{align*}
		\Ical_{\Rbb^2 \times V} = \left\{A \to B \biggl| \begin{array}{l} A \subseteq \Rbb^2 \times V \text{ and } \pi_{\Rbb^2}(B) \subseteq \gamma_{\Rbb^2}\circ\pi_{\Rbb^2}(A) \text{ and }\\ \pi_{V}(B) \in \binom{V}{1}\cup V:  \pi_{V}(A)=  \pi_{V}(B) \text{ if } \Pi_V(A)=\# 1,  \pi_{V}(B)= V \text{ else} \end{array} \right\}
	\end{align*}
	with $\pi_{\mathbb{R}^2}: \mathbb{R}^2 \times V \to \mathbb{R}^2$ being the projection onto $\mathbb{R}^2$. Similarly, we set $\pi_V$. 
\end{example}
 
From the definition of a closure system the definition of the family of implications is straight forward. Reverse, one can obtain a closure system based on a family of implications as follows.
\begin{definition}
Let $\Ical_G$ be a family of implications. We say that $D \subseteq G$ \textit{respects an implication} $A \to B$ if and only if either $A \not\subseteq D$ or $A \subseteq D$ then $B \subseteq D$ also follows. We set $\Ecal_{\Ical_G} = \{D \subseteq G \mid D \text{ respects every implication in }\Ical_G\}.$ 
\end{definition}

As this definition already suggests, there is a one-to-one correspondence between closure systems/operators and the set of all closed families of all implications:

\begin{lemma}
	Let $\Ecal_G$ be a closure system and $\Ical_G$ the family of all implications that respect $\Ecal_G$. Then $\Ical_G$ is unique and $\Ecal_G = \Ecal_{\Ical_G}$. In particular, this then states that $\Ical_G$ uniquely defines a closure system.
\end{lemma}

\begin{proof}
	The uniqueness follows directly.
	For the second part, assume in contradiction that $\Ecal_G \neq \Ecal_{\Ical_G}$. In the first case, we assume that $E \in \Ecal_G \setminus \Ecal_{\Ical_G}$. Since $E \not\in \Ecal_{\Ical_G}$, there exists an implication $A \to B \in \Ical_G$ with $A \subseteq E$, but $B \not\subseteq E$. But since $\Ical_G$ consists of all implications that hold for $\mathcal{E}_G$ this implies that $E \not \in \mathcal{E}_G$. This contradicts the assumption. 
	
	For the reverse, assume that $E \in \Ecal_{\Ical_G}\setminus \Ecal_G$. This means that for all implications $A \to B \in \Ical_G$, if $A \subseteq E$, then $B \subseteq E$ is also true. Since $E \not\in \Ecal_G$ we get that $g \in \gamma(E)\setminus E$. However, this implies that $E \to g$ is an implication that holds in $\Ecal_G$ and therefore should lie in $\Ical_G$. So $E \not\in \Ecal_{\Ical_G}$, which is a contradiction, and we obtain the claim.
\end{proof}

Before we continue, let us take a closer look at the set of implications $\Ical_G$. 
We can immediately see that some implications follow semantically from others. For example, if $A_1 \to A_2$ and $B \supseteq A_1$, then we get $B \to A_2$. So the implication $B \to A_2$ is somewhat redundant, since it follows from $A_1 \to A_2$. These semantic structures are summarized by \cite{maier83} as inference axioms, see \cite[p. 45]{maier83}: Let $A,B,C,D, A_1, A_2, B_1, B_2 \subseteq G$. Then we say that the axiom of \textit{reflexivity} holds iff $A \to A$, the axiom of \textit{augmentation} holds iff $A_1 \to B$ implies $A_1 \cup A_2 \to B$, the axiom of \textit{additivity} holds iff $A \to B_1$ and $A \to B_2$ imply $A \to B_1 \cup B_2$, axiom of \textit{projectivity} holds iff $A \to B_1 \cup B_2$ implies $A \to B_1$, axiom of \textit{transitivity} holds iff $A \to B$ and $B \to C$ imply $A \to C$, and the axiom of \textit{pseodotransitivity} holds iff $A \to B$ and $B \cup C \to D$ imply $A \cup C \to D$.

Armstrong proved, see \cite{armstrong74}, that the iterative repetition of these inference axioms on a set of implications (on a set $G$) leads to a family of implications that equals the set of all implications that hold for a closure system on $G$.
Note, however, that when deleting implications that follow from others, one may delete too many implications and end up not representing the same closure system, see Section~3.2.~of the main article. Therefore, we say that a family of implication $\Ical_G$ is complete iff every implication that holds for a closure system follows semantically from $\Ical_G$, see \citep[p. 81]{ganter12}:

\begin{definition}
Let $\Ecal_G$ be a closure system with corresponding closure operator $\gamma_G$ and $\Ical_G$ a family of implications. Then $\Ical_G$ is \textit{complete} w.r.t $\Ecal_G$ if and only if $\Ecal_G = \Ecal_{\Ical_G}$.
\end{definition}

\begin{remark}
	Finally, we want to point out that everything, the closure system, the implications and later the ufg-depth, depends on the application of a reasonable scaling method. The closure system and the implications provide a tool for analyzing/discussing the relational structure in detail, but the starting point is the scaling method. In particular, all the underlying assumptions of the ground space structure are determined by the scaling method. For example, consider the \textsc{Gorillas} example in Section~2 of the main article. There the ground space is $\Rbb^2\times V$. Therefore, if we have two observations in the same place with the same vegetation, we assume them to be duplication of the same objects. Another approach, not discussed here, is to consider each individual nesting point as an observation that cannot be a duplicate, but is another object in the ground space. In this way we can observe more than one object at the same place with the same vegetation. The main article sticks to the first perspective given in Example~1 and 2. 
\end{remark}

\section{Claims and Proofs}
In this section we present the proofs for the claims made in the main part that do not have a reference to the literature containing the proof. We divided the claims into the corresponding sections in the main article. Since we discuss further lemmas to show the claims, the enumeration of lemmas, theorems, etc. differs from that in the main article.

\subsection*{Claims and Proofs of Section 4 - The Union-Free Generic Depth}

First of all, we consider a general observation for ufg-implications. 

\begin{lemma}\label{lem: equivalent ufg def}
	Let $\Kbb = (G, M, I)$ be a formal context. Then we have $A \in \Ical^{prem}_{ufg}$ if and only if there exists $b \in \gamma(A)$ such that for all $a \in A$ and all $\tilde{a} \in A \setminus a$ exists $m \in \Psi(A \setminus a)$ with $(\tilde{a}, m) \in I$ and $(b,m), (a,m) \not\in I$. In other words, $A$ is an ufg-premise if and only if there exists an element in the conclusion where every element in $A$ is needed.
\end{lemma}
\begin{proof}
	The claim that the second statement implies the first statement follows directly by the definition of the ufg-premise. For the reverse, assume that $A$ is an ufg-premise. Then for $A_g = A \setminus \{g\}$ with $g \in A$ we have by Condition (C2) that $g \in \gamma(A) \setminus \cup_{g \in A} \gamma(A_g)$ which is exactly the second statement.
\end{proof}

In Section~4 of the main part of the article, we formalize that the triangles in $\mathbb{R}^2$ together with the convex closure operator define indeed the set of ufg-implications based on formal context $\mathbb{K}_{\mathbb{R}^2}$.
\begin{lemma}\label{lem:ufg set for R^2}
	For $\mathbb{K}_{\mathbb{R}^2}$, the spatial formal context of Example~1 in Section~2, we have $\mathcal{E}_{\mathcal{I}_{\mathbb{R}^2, \text{ufg}}} = \mathcal{E}_{\mathcal{I}_{\mathbb{R}^2}}$. Moreover, $\mathcal{I}_{\Rbb^2, \text{ufg}}$ is the reduced version of $\mathcal{I}_{\Rbb^2}$ without the implications following from the Armstrong rules of reflexivity, augmentation, additivity, and projectivity.
\end{lemma}

\begin{proof}
	First, we prove $\Ecal_{\Ical_{\Rbb^2}} = \Ecal_{\Ical_{\Rbb^2, \text{ufg}}}$. Since $\Ical_{\Rbb^2, \text{ufg}} \subseteq \Ical_{\Rbb^2}$, we have that if $D \subsetneq \mathbb{R}^2$ respects all implications in $\Ical_{\Rbb^2}$, then it also respects all implications in $\Ical_{\Rbb^2, \text{ufg}}$. Therefore, $\Ecal_{\Ical_{\Rbb^2}} \subseteq \Ecal_{\Ical_{\Rbb^2, \text{ufg}}}$. For the subset relation, let $D \in \Ecal_{\Ical_{\Rbb^2, \text{ufg}}}$ and $A \to \gamma_{\Rbb^2}(A)$ be an arbitrary implication in $\Ical_{\Rbb^2}\setminus (\Ical_{\Rbb^2, \text{ufg}} \cup \{A \to A \mid \# A = 1\})$. By Carath{\'e}odory's theorem, see~\cite{eckhoff93}, we get that for every $g \in \gamma_{\Rbb^2}(A)$ there exist $a_1^g, a_2^g, a_3^g \in A$ such that $g \in \gamma_{\Rbb^2}(\{a_1^g, a_2^g, a_3^g \})$. In particular, $\{a_1^g, a_2^g, a_3^g \} \to \gamma_{\Rbb^2}(\{a_1^g, a_2^g, a_3^g \}) \in \Ical_{\Rbb^2, \text{ufg}}$. Since $D$ respects all implications in $\Ical_{\Rbb^2, \text{ufg}}$, $\gamma(A) \supseteq \bigcup_{g \in \gamma_{\Rbb^2}(A)} \{a_1^g, a_2^g, a_3^g \}$ and $\bigcup_{g \in \gamma(A)} \gamma_{\Rbb}(\{a_1^g, a_2^g, a_3^g \}) = \gamma_{\Rbb^2}(A)$ we get that $D$ also respects $A \to \gamma_{\Rbb^2}(A)$. So we have $\Ecal_{\Ical_{\Rbb^2}} = \Ecal_{\Ical_{\Rbb^2, \text{ufg}}}$.

	Now we have to show that $\Ical_{\Rbb^2, \text{ufg}}$ does not contain any further implications that follow from reflexivity, augmentation, additivity, and projectivity. Since the conclusion is set to $\gamma_{\Rbb^2}(A)$ for some premise $A$, we get that there cannot be a proper superset of $\gamma_{\Rbb^2}(A)$ such that the implication holds for the convex sets. A similar argument provides that it is reduced for the additivity rule. The reflexivity and augmentation follow from the fact that we consider non-degenerate simplices together with Carath{\'e}odory's Theorem.
\end{proof}

Moreover, we provide generally that the maximal cardinality of an ufg-premise is bounded by the VC-dimension of the extent set of the formal context. 
\begin{lemma}\label{lem: vc as upper bound of J}
	Let $\Kbb$ be a formal context that has a unique ufg-family of implications $\Ical_{G}$. Let $\vc$ be the VC dimension of the extent sets. Then $\max \{\#A \mid A \in  \Ical_{G, \text{ufg}}^{prem}\} \le \vc$.
\end{lemma}
\begin{proof}
	This proof is a slight adaptation of the proof given in~\cite{blocher23}, Theorem~4. To prove $\max \{A \in  \Ical_{G, \text{ufg}}^{prem}\} \le \vc$ take an arbitrary subset $Q=\{g_1,\ldots, g_k\}$ and ufg-premise of size $k> vc$. Then this subset is not shatterable because $vc$ is the largest cardinality of a shatterable set. Thus, there exists a subset $R\subseteq Q$ that cannot be obtained as an intersection of $Q$ and some $\gamma(A)$ with $A \subseteq G$. In particular, this holds for $R = A$. Thus, $R \neq \gamma(R) \cap Q$ and with the extensitivity of $\gamma$ we get $R \subsetneq \gamma(R)\cap Q$. This means that there exists an object $\tilde{g}$ in $\gamma(R) \cap Q \backslash R$ for which the formal implication 
$R\rightarrow \{\tilde{g}\}$ holds. Thus, (because of the Armstrong rules, cf., \cite[p. 581]{armstrong74}) the object $\tilde{g}$ is redundant in the sense of $Q\backslash\{\tilde{g}\} \rightarrow Q$ and thus $Q$ is not minimal with respect to $\gamma$. Therefore, $Q$ is not an ufg-premise which completes the proof.
\end{proof}

\subsection*{Claims and Proofs of Section 5 - Structural Properties}
Section~5 in the main article discusses the structural properties of depth functions using formal concept analysis given by~\cite{blocher23b}. Here, we provide the proofs to the claims done in the main article.

\begin{theorem}
	Let $P, \tilde{P} \in \Pcal_G$ be two probability measures on $G$ and let $\Kbb, \tilde{\Kbb} \in \varkappa$ be two formal contexts on $G$. \\
	\underline{Invariance on the extents:} Assume that there exists a bijective and bimeasureable function $i: G \to G$ such that the extents are preserved (i.e.~$E$ extent w.r.t. $\Kbb$ $\Leftrightarrow i(E)$ extent w.r.t. $\tilde{\Kbb}$) and the probability as well (i.e. $P(E) = \tilde{P}(i(E))$). Then $
		D_G(g, \Kbb, P) \le D_G(\tilde{g}, \Kbb, P) 
		\Leftrightarrow \tilde{D}_G(i(g), \tilde{\Kbb}, \tilde{P}) \le \tilde{D}_G(i(\tilde{g}), \tilde{\Kbb}, \tilde{P}) $
	is true for all $g, \tilde{g} \in G$.\\
	\underline{Invariance on the attributes:} Let $g_1, g_2 \in G$ with $\Psi_{\Kbb}(g_1) = \Psi_{\Kbb}(g_2)$, then $D(g_1, \Kbb, P) = D(g_2, \Kbb, P)$ holds.
\end{theorem}
\begin{proof}
	Observe that the ufg-depth is based on the extent set. Thus, if two formal contexts result in the same extent set and the probability measure is also preserved by a function $i$, then the ufg-depth does not change. For the invariance on the attributes we use that for every $E \in \mathcal{E}$ we have $g_1 \in E$ iff $g_2 \in E$. So $g_1$ is in an ufg-conclusion iff $g_2$ is in the ufg-conclusion and therefore the ufg-depths must be equal.
\end{proof}

\begin{theorem}
	Let $P \in \Pcal_G$ and formal context $\Kbb \in \varkappa$ with $g_1, g_2 \in G$ such that $\gamma_{\mathbb{K}}(\{g_1\}) \supseteq \gamma_{\mathbb{K}}(\{g_2\})$. Then the \underline{isotonicity property} $D(g_1, \Kbb, P) \le D(g_2, \Kbb, P)$ is true.
\end{theorem}
\begin{proof}
	This follows immediately from the fact that for every ufg-premise $U$ with $U \to g_1$ we have $U \to g_2$. So the probability of ufg-conclusions containing $g_1$ is a smaller than of those containing $g_2$, which provides the claim.
\end{proof}

\begin{theorem}
	Let $D(\cdot, \Kbb, P)$ be a depth based on formal concept analysis. Then $D^{qc}$ is a quasiconcave function.
\end{theorem}
\begin{proof}
	 We show that for every $\alpha \in \Rbb$ $\gamma(Cont_{D^{qc}, \alpha}) = Cont_{D^{qc}, \alpha}$ is true. Assume in contradiction that there exists $\alpha\in \Rbb$ such that $g \in \gamma(Cont_{D^{qc}, \alpha})\setminus Cont_{D^{qc}, \alpha}$.
Then $Cont_{D^{qc}, \alpha} \to g$ is a valid implication. 

Case 1: $Cont_{D^{qc}, \alpha} = \emptyset$. Then we know that every subset $A \subseteq G$ implies $g$. Hence, for every $\alpha$ we get that $Cont_{D, \alpha} \to g$ is true and therefore $D^{qc}(g, \Kbb, P)$ has a maximum depth value. So it can never contradict the quasiconcavity assumption by having a depth value that is too small.

Case 2: $Cont_{D^{qc}, \alpha} \neq \emptyset$. Since $D^{qc}(g, \Kbb, P) < \alpha $, we get that there exists $\varepsilon >0$ such that for every $\alpha' > \alpha - \varepsilon $ it holds that $Cont_{D, \alpha'} \not\to g$ is true, otherwise the depth of $D^{qc}(g, \Kbb, P)$ must be at least $\alpha'$. However, the construction of the quasiconcave depth function gives us that for every $\alpha' < \alpha$ we have $Cont_{D, \alpha'} \to Cont_{D^{qc}, \alpha}$. So by the inference axioms we know that $Cont_{D, \alpha'} \to g$ is also valid for every $\alpha - \varepsilon < \alpha' < \alpha$ which is a contradiction.
\end{proof}

\begin{theorem}
	Let $\Kbb \in \varkappa$. Let $g_1, \ldots, g_n$ be a sample of $G$ with $n \in \mathbb{N}$. We denote with $P^{(n)}$ the empirical probability measure given by $g_1, \ldots, g_n$ and by $P^{(n, -\ell)}$ the empirical probability measure based on $g_1, \ldots, g_{\ell-1},  g_{\ell+1}, \ldots,  g_n$ with $\ell \in \{1, \ldots, n\}.$ \\
	\underline{Respecting duplication}: Let $i,\ell \in \{1, \ldots, n\}$ with $i \neq \ell$ and for every extent $E \in \mathcal{E}$ we have $g_{\ell} \in E$ iff $g_{i} \in E$. Moreover, assume that there exists $j \in J_{P, \Kbb}$ and ufg-premises $A_1, A_2 \in 2^{\{g_1, \ldots, g_{\ell-1},  g_{\ell+1}, \ldots,  g_n\}} \cap \Ical_{ufg}^{prem,j}$ with $g_i \in A_1$ and $g_i \not\in \gamma(A_2)$. Then, we have $D_G(g_i, \Kbb, P^{(n, -\ell)}) < D_G(g_i, \Kbb, P^{(n)}).$\\
		\underline{Stability of the order}: Assume that the only extents $E$ that contains $g_{\ell}$ for $\ell \in \{1, \ldots, n\}$ as well as any subset of $g_1, \ldots, g_{\ell-1},  g_{\ell+1}, \ldots$ is $E = G$. Then for $g, \tilde{g} \in \{g_1, \ldots, g_{\ell -1}, g_{\ell + 1}, \ldots, g_n\}$ we have $D_G(g, \Kbb, P^{(n)}) \le D_G(\tilde{g}, \Kbb, P	^{(n)}) 
		\Leftrightarrow \tilde{D}_G(g, \Kbb, P	^{(n, -\ell)}) \le \tilde{D}_G(\tilde{g}, \Kbb, P	^{(n, -\ell)}). $
\end{theorem}

\begin{proof}
The assumption of the respecting duplication property implies that there are $g_i, g_{\ell}$ in the sample with $i \neq j$ and $\Psi(g_i) = \Psi(g_{\ell})$. We assume that in the full sample there exists a further ufg-premise containing $g_i$ than in the reduced sample. Since the proportion is not already one, this provides the claim.

Now, we show the stability of the order property. Let $g_\ell$ be an element of the sample that is completely different to the rest. Then for every ufg-premise   $A \in 2^{\{g_1, \ldots,  g_n\}} \cap \Ical_{ufg}^{prem}$ we have that either $g_\ell \in A$ and the entire sample lies in the conclusion or $g_\ell \not\in A$. Thus, for every $j \in J_{P, \Kbb}$ and every $g \in \{g_1, \ldots, g_{\ell-1},  g_{\ell+1}, \ldots,  g_n\}$ the amount added in the proportion equals when including the observation $g_\ell$.

\end{proof}

\begin{remark}
	Note that the assumptions on the existence of the two ufg-implications in the invariance on the extents property are indeed necessary. The first assumption, that there exists an implication $A_1 \to \gamma(A_1)$ with $g_i \in A$, is necessary because otherwise this object $g_i$, and hence also object $g_{\ell}$, has no effect on the ufg-depth. Note that this assumption is generally true, and in particular holds for all the examples discussed in the main article.
	
	The second assumption, that there is an ufg-premise $A_2 \to \gamma(A_2)$ with $g_i \not\in \gamma(A_2)$, ensures that the proportion does indeed increase. If there is no such ufg-implication, then the ufg-depth of $g_i$ is already maximal and therefore cannot increase. (Note that this property has a strictly larger in its definition).
\end{remark}

\begin{theorem}
Let $\Kbb \in \varkappa, P \in \Pcal$ and $J_{P, \Kbb}\subseteq \Nbb$ be the same as in Definition~4.2. of the main article. Let $X_1, \ldots, X_n\overset{i.i.d.}{\sim} P$. We assume that $\#J_{P, \Kbb}< \infty$. Moreover, we assume that for every $j \in J_{P, \Kbb}$ $\Ical^{concl, j}_{\Kbb,\text{ufg}}$ has finite VC-dimension. 
With this, we get the \underline{consistency property} with $\sup_{g \in G} |D(g, \Kbb, P^{(n)}) - D(g, \Kbb, P)| \to 0$ almost surely for $n$ to infinity. (We assume that this supremum is measurable.)
\end{theorem}
\begin{proof}
Recall the notation defined before Definition~4.2. of the main article. The proof can be divided into three parts. First, we proof that for all $j \in J_{P, \Kbb}: \: \sup_{g \in G} |U^j_{(X_1, \ldots X_n)}[f_g^j] - \Ebb[f_g^j] \mid \to 0$ almost surely. Second, we show $|\frac{C_j}{U^j_{(X_1, \ldots X_n)}[h^j]} - \frac{C_j}{\Ebb[h^j]} \mid \to 0$ almost surely for every $j \in J_{P, \Kbb}$. In the last part, we combine the first two to provide the claim.

Part 1: For $j \in J_{P, \Kbb}$ consider the dual set $\tilde{\Fcal}^j$. Since $\Ical^{concl, j}_{\Kbb, \text{ufg}}$ has a finite VC-dimension, the sub-graphs of $\tilde{\Fcal}^j$ have also a finite VC-dimension. By~\cite{assouad83} we obtain that the sub-graphs of $\Fcal^j$ have also finite VC-dimension. With~\cite{arcones1993limit} we get the first part.

Part 2: Let $j \in J_{P, \Kbb}$. We use Theorem 2.3. of~\cite{christofides92} to obtain $| U^j_{(X_1, \ldots X_n)}[h^j] - \Ebb[h^j]| \to 0$ almost surely. In particular, this also implies that $U^j_{(X_1, \ldots X_n)}[h^j] $ is almost surely positive for $n$ large enough (note that $\mathbb{E}[h^j]>0$ for $j \in J_{P,\Kbb}$). 
Hence, the function $x \to 1_{x >0} 1/x $ is almost surely evaluated only at a positive argument, if $n$ is large enough. Therefore, because this function is continuous for positive arguments we obtain Part~2. 

Part 3: For $j \in J_{P, \Kbb}$ we consider the following inequality, which uses that function $h^j$ is independent of $g \in G$, a decomposition of the factors and the triangle inequality.
\begin{align*}
& \sup_{g \in G} \biggl|\frac{C_j}{U^j_{(X_1, \ldots X_n)}[h^j]} U^j_{(X_1, \ldots X_n)}[f_g^j]- \frac{C_j}{\Ebb[h^j]}\Ebb[f_g^j] \biggl| \\
\le & \sup_{g \in G} \biggl| U^j_{(X_1, \ldots X_n)}[f_g^j]\biggl|\:\biggl|\frac{C_j}{U^j_{(X_1, \ldots X_n)}[h^j]} - \frac{C_j}{\Ebb[h^j]} \biggl|
+ \biggl| \frac{C_j}{\Ebb[h^j]}\biggl|  \sup_{g \in G}\biggl|U^j_{(X_1, \ldots X_n)}[f_g^j] - \Ebb[f_g^j] \biggl|  
\end{align*}

Since the first two components of the multiplications can be bounded by above for every $g$ and every (empirical) probability measure. We obtain with Part 1 and 2 and $J_{P, \Kbb}$ being finite the claim.

\end{proof}
 
 \begin{theorem}
 	Let $\Kbb = (G, M, I)$ be a formal context and $J_{P, \Kbb}$ be finite with $\sup\{j \mid P \in \mathcal{P}, j \in J_{P, \Kbb}\} = K \in \Nbb$. Additionally, we assume that $\Kbb$ satisfies the following further conditions:  
 	\begin{enumerate}
 		\item[(A1)] For all $A \subseteq G$ and all $B \subseteq G\setminus \gamma(A)$ with $B$ finite there exists an extent $S \subseteq G\backslash B$ such that for all $g \in A$ there exists an ufg-premise $U \subseteq S$ with $U \rightarrow \{g\}$ .
 		\item[(A2)] There exists $L > 0$ such that  $\frac{\max_{j \in J_{P, \Kbb}} C_j / \mathbb{E}[h^j]}{\min_{j \in J_P} C_j / \mathbb{E}[h^j]} \le L$ for every $P \in\mathcal{P}$.
 	\end{enumerate}	 
 	Then \underline{$D^{qc}$ is approximately weakly free w.r.t. quasiconcavity} in the following sense: For every quasiconcave depth function $E$ on $\Kbb$, for every probability measure $P \in \mathcal{P}$ and every finite $\tilde{G} \subseteq G$ there exists a probability measure $P^*$ on $G$ (with finite support) such that for all $g, \tilde{g} \in \tilde{G}$ we have
 	\begin{align}\label{ordering}
 		E(g, \Kbb, P) > E(\tilde{g}, \Kbb, P) \Rightarrow (D_{|\tilde{G}})^{qc}(g, \Kbb, P^*) > (D_{|\tilde{G}})^{qc} (\tilde{g}, \Kbb, P^*).
 	\end{align}
 \end{theorem}
 
 \begin{proof}
 	The proof is divided into two parts. First we define the probability measure $P^*$ and show that the depth values $D(g)$ of the objects in $\tilde{G}$ w.r.t. the original ufg-depth $D$ satisfy (\ref{ordering}). In the second step, we consider the depth values of the quasiconcave version $D^{qc}$ of $D$ and show that they also fulfill property (\ref{ordering}).\vspace{1em}
 	
 	\underline{Part 1:} Let $e_{(1)}, \ldots, e_{(k)}$ be the increasingly ordered (unique) depth values of $E(\tilde{G}, \Kbb, P)$ and let $G_{(i)} = \{g \in \tilde{G} \mid E(g, \Kbb, P) = e_{(i)}\}$.
 	
 	Now, we go step by step through the layers $G_{(i)}$ given by $E$. First we set $P^{(k+1)}= 0$ for all $g \in G$ and will modify it in the following process. \vspace{1em}
 	
 	\textit{Step 1:} We start with layer $G_{(k)}$ corresponding to the highest value $e_{(k)}$. 
 	We set $A_{(k)} = G_{(k)}$ and $B_{(k)} = \tilde{G} \setminus A_{(k)} \subseteq G \setminus A_{(k)}$. Note that $B_{(k)}$ is finite.
 	Due to assumption (A1) there exists an extent $S_{(k)} \subseteq G \setminus B_{(k)}$ such that for all $g \in A_{(k)}$ there exists an ufg-premise $U_g \subseteq S_{(k)}$ with $U_g \rightarrow g$. Since $S_{(k)}$ is an extent, we know that no element of $B_{(k)}$ is implied by $U_g$. We set $U_{(k)} = \bigcup_{g \in A_{(k)}} U_g$, $\mathcal{U}_{(k)} = \{U_g \mid g \in A_{(k)}\}$. 
 	\vspace{1em}

 	\textit{Step 2:} Now, we proceed with $G_{(k-1)}$. Similar to Step~1 we set $A_{(k-1)} = G_{(k)} \cup G_{(k-1)}$ and $B_{(k-1)} = \tilde{G} \setminus A_{(k-1)}$. Again, Assumption (A1) provides us a set $S_{(k-1)} \subseteq G \setminus B_{(k-1)}$ such that for every $g \in A_{(k-1)}$ exists an ufg-premise $U_g \subseteq S_{(k-1)}$ with $U_g \rightarrow g$ but $U_g$ implies no element in $B_{(k-1)}$. Note that for every such $U_g$ there exists at least one $u \in U_g$ with $u \not\in U_{(k)}$, because otherwise we would have $U_g \rightarrow g$. 
 	We set $U_{(k-1)} = \bigcup_{g \in A_{(k-1)}} U_g$, $\mathcal{U}_{(k-1)} = \{U_g \mid g \in A_{(k-1)}\}$.
 	\vspace{1em}
 	
 	\textit{Step 3 to $k-1$}. We proceed similar until we defined $\mathcal{U}_{(1)}$ and $A_{(1)}$.\vspace{1em}

 	\underline{Part 2:}
 	Now, we set $U = \bigcup_{i = 1}^k U_{(i)}$, $\mathcal{U} = \bigcup_{i = 1}^k \mathcal{U}_{(i)}$ and $c:= L\cdot \# \mathcal{U} \cdot k$. Note that $c$ does not depend on the probability measure $P^*$ that we will now construct. Based on the above definitions, we can now define again successively the probability measure $P^*$ which will have support $U$:\vspace{1em}
 	
 	\textit{Step a:}
 	We define $p_{(k)}$ and $P^{(k)}$ corresponding to the contour set of $E$ having the highest depth value $e_{(k)}$.
 	\begin{align*}
 		&0 < p_{(k)} \le  \frac{1}{\# U \cdot k}\\
 		&P^{(k)}(g) := \begin{cases} p_{(k)}  , &g \in U_{(k)} \\
 			0 &else \end{cases}
 	\end{align*}
 	
 	\textit{Step b:}
 	We define $p_{(k-1)}$ and $P^{(k-1)}$ corresponding to the contour set of $E$ having the second highest depth value $e_{(k)}$.
 	\begin{align*}
 		&0 < p_{(k-1)} < p_{(k)} \text{ such that } p_{(k)}^K > c  \cdot p_{(k-1)} \\
 		&P^{(k-1)}(g) := \begin{cases} p_{(k-1)} , &g \in U_{(k-1)}\\
 			0, &else \end{cases}
 	\end{align*}
 	
 	\textit{Step c and on:} Analogously to Step~b we define  $p_{(i)}$ and $P^{(i)}$ for $i \in \{1, \ldots k-2\}$. Note that also for arbitrary $i,j$ with $i<j$ we have $p_{(j)}^K > c\cdot  p_{(i)}$. 
 	\vspace{1em}
 	
 	Finally, we set $P(g) = \sum_{i =1}^k P^{(i)}(g)$ for all $g \in U$ and with this we obtain that 
 	\begin{align*}
 		0 <   \sum_{g \in U} P(g) \le \sum_{g \in U}\left( \sum_{i = 1}^k  P^{(i)}(g)\right)\le \sum_{g \in U}\left( \sum_{i = 1}^k  \frac{1}{\# U \cdot k}\right)\le 1.
 	\end{align*}
 	Thus, we get $0 \le \varepsilon := 1 - \sum_{g \in U} P(g) \le 1$ and with this we can now define the probability measure $P^*$:
 	\begin{align*}
 		P^*(g) = \begin{cases} P(g) + \frac{\varepsilon}{\# U_{(k)}}, &g \in U_{(k)} \\
 			P(g), & else\end{cases}.
 	\end{align*}
 	With this, $P^*$ defines a probability measure on $G$.\vspace{1em}
 	
 	Let us take a look at the upper and lower bounds of the ufg-depth $D(g, \Kbb, P^*)$ for $g \in \tilde{G}$. Let $i \in  {{\{1,\ldots, k}}\}$ and $g \in G_{(i)}$. We know that there exists at least one ufg-premise $U$ in {$\mathcal{U}_{(i)}$} 
 	with maximal cardinality $K$ and with $U\rightarrow \{g\}$. Thus, we have as lower bound for $g \in {G}_{(i)}$
 	\begin{align*}
 		D(g, \Kbb, P^*) \ge {{\frac{C_i}{\mathbb{E}\left[h^i\right]}}} p_{(i)}^{K}
 	\end{align*}
 	
 	For $i \in\{1,\ldots, k-1\}$ and $g \in G_{(i)}$, we get as an upper bound:
 	\begin{align*}
 		D(g, \Kbb, P^*) \le \max\limits_{j \in J}\frac{C_j}{\mathbb{E}[h^j]}\# \mathcal{U}\cdot k  \cdot p_{(i)}
 	\end{align*}
 	
 	With this, for $i,\ell \in\{1,\ldots,k\}$ with $i<\ell$ and $g \in G_{(i)}$ and $ \tilde{g}\in G_{(\ell)}$ we immediately get
 	\begin{align*}D(g,\mathbb{K},P^*) &\leq \max\limits_{j \in J}\frac{C_j}{\mathbb{E}[h^j]} \cdot \# \mathcal{U}\cdot k\cdot p_{(i)}\\
 		&\leq L\cdot \frac{C_\ell}{\mathbb{E}[h^\ell]} \cdot \# \mathcal{U}\cdot k \cdot p_{(i)}\\
 		&= \frac{C_\ell}{\mathbb{E}(h^\ell)} \cdot c \cdot p_{(i)}\\
 		& \leq \frac{C_\ell}{\mathbb{E}(h^\ell)} \cdot p_{\ell}^K \leq D(\tilde{g},\mathbb{K},P^*).
 	\end{align*}

 	Until now, we showed that for $g,\tilde{g} \in \tilde{G}$ we have $E(g,\Kbb,P) > E(\tilde{g},\Kbb,P)    \Longrightarrow D(g, \Kbb,P^*)  > D(\tilde{g},\Kbb,P^*)$. Now we show that the same holds for $\left(D_{\mid \tilde{G}}\right)^{qc}$:   
 		 Let $i,j \in \{1, \ldots, k\}$ with $j < i$ and let $g_i \in G_i , g_j \in G_j$. Then we know that
 	\begin{align*}
 		D(g_i, \Kbb, P^*) > D(g_j, \Kbb, P^*).
 	\end{align*}
 	Assume in contradiction that $(D_{\mid \tilde{G}})^{qc}(g_i, \Kbb, P^*) \le (D_{\mid \tilde{G}})^{qc}(g_j, \Kbb, P^*)$. Then there exists a set $A \subseteq G_{(j+1)}\cup \ldots G_{(k)}$ such that the implication $A \rightarrow g_j$ with $D(g, \Kbb, P^*) \geq D(g_i, \Kbb, P^*)$ for all $g \in A$ holds. But this implication is in clear contradiction with the quasiconcavity of $E$, because $E(g,\Kbb,P) > E(g_j,\Kbb,P)$ for all $g \in A$. 
 	 Therefore, in fact we have $ E(g,\Kbb,P) > E(\tilde{g},\Kbb,P) \Longrightarrow \left(D_{\mid\tilde{G}}\right)^{qc}(g,\Kbb,P^*) > \left(D_{\mid\tilde{G}}\right)^{qc}(\tilde{g},\Kbb,P^*) $ for all $g,\tilde{g}\in \tilde{G}$. 
 	
 \end{proof}
 
 \begin{remark}
 	In general, Assumptions $(A1)$ and $(A2)$ are very strong. Assumption $(A1)$ can be seen as a separability condition. It is satisfied, e.g., in the case of $\mathbb{R}^d$ together with the conceptual scaling discussed in Example~1 of the main article. Assumption~$(A1)$ is also true if every one-element set is an ufg-premise, because in this case one can set $U_g:=\{g\}$. For example, this condition hods for the case of hierarchical-nominal data if duplication is allowed. 
 	
 Assumption~$(A2)$ is generally difficult to satisfy. If $G$ is finite, Assumption~$(A2)$ cannot be satisfied at all. However, a slightly modified choice of the coefficients $C_j$ as $C_j:= \mathbb{E}[h^j]/(\varepsilon + \mathbb{E}[h^j]) \approx1$ with some small fixed constant $\varepsilon>0$ makes Assumption $(A2)$ satisfied. Beyond that, in the case of $\mathbb{R}^d$ with the scaling method in Example~1 of the main article and an absolutely continuous probability measure, we have that $\mathbb{E}\left[h^j\right]$ is one if for $j\in \{2,\ldots, d+1\}$ and zero otherwise. Of course, when we construct $P^*$ as in the proof, we use a discrete probability measure. Note, however, that for $\mathbb{R}^d$ one can always use points that are affine independent for the construction of the ufg-premises $U_g$ such that, independent of the probability values $p_{(i)}$, we always have $\mathbb{E}\left[h^j\right]=1$ for $j\in \{2,\ldots, d+1\}$ and zero otherwise. Note also that it is sufficient to assume that the boundedness Assumption~$(A2)$ for $P$ holds over all probability measures implicitly used in the construction of $P^*$ (which is of course a hard condition to keep track of). 
 \end{remark}

In the following we consider the formal context $\Kbb = (G, M,I)$ given by hierarchical-nominal data and the scaling method presented in Example~3 of Section~2 in the main article. Analogously we denote the attributes by $x_1x_2\ldots x_k$ with $x_i$ describing the category on level $i$ depending on the previous levels $1, \ldots, i-1$.

\begin{lemma}\label{lem: hierarchisch_ufg_premises}
	Let $\Kbb = (G, M,I)$ be a formal context given by hierarchical-nominal data with $L\ge2$ levels and at each level at least $3$ categories. We use the scaling method presented in Example~3 of Section~2 of the main article. We assume that for each object $g\in G$ there exists another object $\tilde{g} \in G$ with $g \neq \tilde{g}$ and $\Psi(\{g\}) = \Psi(\{\tilde{g}\})$. Then the ufg-premises have cardinality one or two.
\end{lemma}

\begin{proof}
	First, we show that there are ufg-premises of cardinality one and two. Let $g \in G$. By assumption there exists an object $\tilde{g} \in G$ with $g \neq \tilde{g}$ and $\Psi(\{g\}) = \Psi(\{\tilde{g}\})$. So $g \to \gamma(\{g\}) \supseteq \{g, \tilde{g}\}$ is a ufg-implication. Since we have at least two levels, we know that there are two objects $\tilde{g} \in G$ with $\Psi(\{g\}) \neq \Psi(\{\tilde{g}\})$ such that they differ at least at Level $L$. Thus the set $\{g, \tilde{g}\}$ implies all objects that can also be sorted into the $x_1, \ldots x_k$ categories with $k < L$. Thus $\{g, \tilde{g}\}$ is union-free, and since it obviously cannot be reduced without also reducing the conclusion, it is also generic.
	
		Finally, we show that every set $\{g_1, \ldots, g_n\}$ for $n\ge 3$ is not an ufg-premise.
	Let $\{g_1, \dots, g_n\} \to B$ be a valid implication given by $\Kbb$. Then $B$ is a subset of $\Phi\circ\Psi(\{g_1, \ldots, g_n\})$. By constructing the attributes in Example~3 of Section~2 of the main article, we get that $\Psi(\{g_1, \ldots, g_n\})$ are exactly the attributes describing the first $k$ level categories on which all objects $g_1, \ldots, g_n$ agree. These attributes can also be described by only two objects $g_i, g_j$ with $i,g \in \{1, \ldots, n\}$, namely by two objects $g_i,g_j$ that agree up to level $k$, but that disagree on level $k+1$. Therefore, the implication $\{g_i, g_j\} \to B$ is also valid and thus, the implication $\{g_1, \ldots, g_n\}\to B$, is not an ufg implication since the premise is not minimal. 
\end{proof}

\begin{theorem}
	Let $\Kbb = (G, M,I)$ be a formal context given by hierarchical-nominal data with $L\ge2$ levels, $K\geq 3$ categories on each level, and the scaling method presented in Example~3 of Section~2 of the main article.  We assume that for each object $g\in G$ there exists another object $\tilde{g} \in G$ with $g \neq \tilde{g}$ and $\Psi(\{g\}) = \Psi(\{\tilde{g}\})$.
	
	 We set $C_1,C_2>0$ in the ufg-depth definition. Then the quasiconcave version \underline{$D^{qc}$ of the} \underline{ufg-depth is strongly free with respect to the property quasiconcavity}. This means that for every $\varepsilon >0$ there exists a family $\mathcal{P}^\varepsilon $ of probability measures with diameter\footnote{The diameter of a family $\mathcal{P} $ of probability measures on a measurable space $(G,\Sigma)$ is defined as $diam(\mathcal{P}):= \sup\limits_{A\in \Sigma; P,Q\in \mathcal{P}} | P(A) - Q(A)|$} less than or equal to $\varepsilon$ such that for any other arbitrary quasiconcave depth function $E$ and any arbitrary probability measure $P$ there exists a measure $P^* \in \mathcal{P}^\varepsilon$ such that
	$$\forall g,\tilde{g} \in G: E(g,\mathbb{K},P) >E(\tilde{g},\mathbb{K},P) \Longrightarrow D^{qc}(g,\mathbb{K},P^*) >D^{qc}(\tilde{g},\mathbb{K},P^*).$$
\end{theorem}

\begin{proof}
	First, we introduce some notation for simplicity. By assumption, we have on the finest category-level $N = K^L$ categories. Moreover, due to the scaling method, we can divide $G$ into $N$ subsets $G_1, \ldots, G_{N}$, where  each $G_i$ corresponds to a set of objects with the same category on the finest level. This means that $G_i \cap G_j = \emptyset$ for different $i,j \in \{1, \ldots, N\}$ and $G_1 \cup \ldots \cup G_N = G$ is true. In particular, there exists an attribute $x_i= {(x_{i})}_1{(x_{i})}_2\ldots {(x_{i})}_N$ such that $G_i = \Phi(\{x_i\})$. So $G_i$ is an extent and for all $g \in G_i, \: \gamma(\{g\}) = G_i$ is true.

	 Let $\varepsilon >0$. We define $$\mathcal{P}^\varepsilon := \left\{P \text{ probability measure} \mid \forall g \in G : P(\{g\}) \in \left[\frac{1/N-\varepsilon/(2N)}{\# \gamma(\{g\})}, \frac{1/N +\varepsilon/(2N)}{\# \gamma(\{g\})}\right]\right\}.$$ Note that $\mathcal{P}^\varepsilon$ has a diameter smaller than or equal to $\varepsilon$. Let $E$ be a quasiconcave depth function. This means that the contour sets are extents. In addition, the contour sets are nested due to their construction. For now on let $\Phi(\{y_1y_2\ldots y_K\})$ be the objects with the largest depth value with respect to $E$. Then the contour sets of $E$ are a subset of the extents $G \supseteq\Phi(\{y_1\})\supseteq \Phi(\{y_1y_2\})\supseteq \ldots \supseteq\Phi(\{y_1y_2\ldots y_K\})$, where $\Phi(y_1\ldots y_K\}$ is the contour set of the objects with the highest depth. Note that $\Phi(\{y_1y_2\ldots y_K\})$ is equal to one of the set $G_i$ corresponding to a division on the finest level. W.l.o.g. we set $\Phi(\{y_1y_2\ldots y_K\}) = G^* = G_N$.
	 
	 Now, we construct $P^* \in \mathcal{P}$ 
	 \begin{align*}
	 	P^*(\{g\})\begin{cases} \frac{1/N + \varepsilon/(2N)}{\# G^*}, & \text{if } g \in G^* \\ \frac{1/N- \varepsilon/(2N\cdot(N-1))}{\# \gamma(g)}, &\text{else} \end{cases}.
	 \end{align*}
	 
	 In the following we show that $D(\cdot, \Kbb, P^*)$ provides the same order as $E$. So we set $y_0 = \emptyset$ and show that for every $\ell \in \{0, \ldots, L-1\}$ and every $g_1 \in \Phi(\{y_1y_2\ldots y_i\}) \setminus \Phi(\{y_1y_2\ldots y_{i+1}\})$ and $g_2 \in \Phi(\{y_1y_2\ldots y_{i+1}\})$, $D(g_1, \Kbb, P^*) < D(g_2, \Kbb, P^*)$ is true. Let $\ell \in \{0, \ldots, L-1\}$ be arbitrary. To obtain the depth function, we need to discuss the ufg implications. From Lemma~\ref{lem: hierarchisch_ufg_premises}, we know that the ufg-premises have either cardinality one or two. 
	 
	\textit{Part 1 - ufg-premises of cardinality one:} Each implication $g \to \gamma(\{g\})$ implies only those objects which have exactly the same attributes. By the definition of $G_1, \ldots, G_N$ there exists $i \in \{1, \ldots, N\}$ such that $g \in G_i$ and $\gamma(\{g\}) = G_i$. This gives us
	\begin{align*}
		&\frac{C_1}{\mathbb{E}[h^1]}\mathbb{E}[f_{g}^1] = \frac{C_1}{\mathbb{E}[h^1]}\mathbb{E}[1_{\gamma(A)}(g)1_{\Kbb^{prem, 1}_{ufg}}(A)] = \begin{cases}  1/N + \varepsilon/(2N), & \text{ if} g  \in G^* \\ 1/N- \varepsilon/(2N\cdot(N-1)), &\text{else} \end{cases}.
	\end{align*}
	
	\textit{Part 2 - ufg-premises of cardinality two:} Let $G_i$ and $G_j$ be such that their corresponding attributes differ for at least on the finest level $L$ (see proof of Lemma~\ref{lem: hierarchisch_ufg_premises}). Then for every $g_i \in G_i$ and $g_j \in G_j$ we have that $\{g_i, g_j\}$ defines an ufg-premise. In particular, with this procedure we obtain all possible ufg-premises of cardinality two, see the proof of Lemma~\ref{lem: hierarchisch_ufg_premises}. We denote all these pairs by $\mathcal{G}^2$. Due to symmetry (on each level we have exactly the same number of categories) we obtain that the number of pairs $G_i,G_j \in \mathcal{G}^2$ such that $g \in \gamma(G_i\cup G_j)$ is the same for all $g \in G$. Let $g \in G$, then we get
	\begin{align*}
		&\frac{C_2}{\mathbb{E}[h^2]}\mathbb{E}[f_{g}^2] = \frac{C_2}{\mathbb{E}[h^2]}\mathbb{E}[1_{\gamma(A)}(g)1_{\Kbb^{prem, 2}_{ufg}}(A)] \\
		&= \frac{C_2}{\mathbb{E}[h^2]}\Big[\mathbb{E}[1_{\gamma(A)}(g)1_{\Kbb^{prem, 2}_{ufg}}(A)1_{\{A \subseteq G\setminus G^*\}}] \\
		& \qquad \qquad\qquad\qquad + \frac{C_2}{\mathbb{E}[h^2]}\mathbb{E}[1_{\gamma(A)}(g)1_{\Kbb^{prem, 2}_{ufg}}(A)1_{\{A \cap G^* \neq \emptyset\}}]\Big]\\
		&= \frac{C_2}{\mathbb{E}[h^2]}\Big[\sum_{\substack{G_i,G_j \in\mathcal{G}^2 \text{ with } \\G_i\neq G^*\neq G_j \text{ and } \\g \in \gamma(G_i \cup G_j)}} (1/N  - \varepsilon/(2N \cdot (N-1)))^2 \\
		&\qquad\qquad\qquad\qquad + \sum_{\substack{G_i,G_j \in\mathcal{G}^2 \text{ with }\\ G_i= G^* \text{ or } G^* = G_j \text{ and }\\ g \in \gamma(G_i \cup G_j)}} (1/N  - \varepsilon/(2N \cdot (N-1))) (1/N + \varepsilon/(2N))  \Big].
	\end{align*}
	 	
If $g_1 \in G^*$ and $g_2 \in G$ the difference between the pairs is that the set of pairs $G_i,G_j$ where at least one is equal to $G^*$ is strictly larger for $g_1$ than for $g_2$. Hence, there are more pairs in the second part of the sum above. With this, we immediately get that $\frac{C_2}{\mathbb{E}[h^2]}\mathbb{E}[f_{g_1}^2] > \frac{C_2}{\mathbb{E}[h^2]}\mathbb{E}[f_{g_2}^2] $.
	 
	 	Now Part~1 and~2 together with the definition of the ufg-depth show that $E$ and $D(\cdot, \Kbb, P^*)$ give the same order of the objects $G$. So $D(\cdot, \Kbb, P^*)$ is already quasiconcave and with $D(\cdot, \Kbb, P^*) = D^{qc}(\cdot, \Kbb, P^*)$ we prove the claim.
\end{proof}

\subsection*{Claims and Proofs of Section 6 - Examples}
Section~6 of the main article discusses two concrete data examples: mixed spatial-categorical-numerical data and hierarchical-nominal data. Here we provide the proof of the claim that simplifies the calculation of the ufg depth for mixed spatial-categorical-numerical data. Therefore, we consider the special case of joined formal contexts. Let us assume that we have two formal contexts on the same object set $G$ but with two different attribute sets, $\Kbb_1 = (G, A_1, I_1)$ and $\Kbb= (G, A_2, I_2)$. Then consider the joined formal context $\Kbb = (G, A_1 \cup A_2, I_1 \cup I_2)$. Analogously we denote the derivation and closure operators. Then for this joint formal context, we get:
\begin{lemma}
Let $\Kbb_1, \Kbb_2$ and $\Kbb_{1,2}$ together with the closure operator $\gamma_1, \gamma_2$ and $\gamma_{1,2}$ be defined as in the beginning of this section. Let $G$ be the set of objects and $A  \subseteq G$. Then $\gamma_{1,2}(A) = \gamma_{1}(A) \cap \gamma_2(A)$ is true.
\end{lemma}

\begin{proof}
	The proof follows from
	\begin{align*}
		a \in \gamma_{1,2}(A)  &\Leftrightarrow \Psi_{1,2}(A) = \Psi_{1}(A) \dot\cup \Psi_{2}(A) \subseteq  \Psi_{1,2}(a)\\
		&\Leftrightarrow a \text{ has every attribute in } \Psi_{1}(A)  \text{ and }a \text{ has every attribute in }  \Psi_{2}(A) \\
		&\Leftrightarrow a \in \gamma_{1}(A)\cap \gamma_2(A).
	\end{align*}
\end{proof}	

\begin{lemma}\label{lem: ufg bound joined fc}
	Let $\Kbb_1, \Kbb_2$ and $\Kbb_{1,2}$ together with the closure operator $\gamma_1, \gamma_2$ and $\gamma_{1,2}$ be defined as in the beginning of this section. Let $\max\{\#A \mid A \in \Ical_{1, ufg}\} = u_1$ and $\max\{\#A \mid A \in \Ical_{2, ufg}\} = u_2$, then every $A \subseteq G$ with $\#A > u_1 + u_2$ cannot be an ufg-premise of $\Kbb$.
\end{lemma}

\begin{proof}
	This proof follows directly from Lemma~\ref{lem: equivalent ufg def}.
\end{proof}

\begin{lemma}
	For the formal context $\Kbb_G$ with extent set given by Equation~(1) in Section 6.1. of the main article, we have for the ufg-family of implications 
	\begin{align*}
		\Ical_{ufg} \subseteq \left\{A \to B \biggl| \begin{array}{l} A \subseteq \Rbb^2 \times V \times \Rbb \text{ and } 2 \le \# A \le 4, \\
		\pi_{\Rbb}(B) = [\min\{\pi_{\mathbb{R}}(A)\},\max\{\pi_{\mathbb{R}}(A)\}], \:  \pi_{\Rbb^2}(B) = \gamma_{\Rbb^2}\circ\pi_{\Rbb}(A), \\
		\pi_{V}(B) \in \binom{V}{1}\cup V:  \pi_{V}(B) =  \pi_{V}(A) \text{ if } \#  \pi_{V}(A) = 1, \pi_{V}(B) = V \text{ else}
		 \end{array} \right\}.
	\end{align*}
\end{lemma}

\begin{proof}
	Note that for $A \subseteq G$ we have $\gamma(A) = \gamma_{\mathbb{R}^2}\circ \pi_{\mathbb{R}^2}(A) \times \tilde{V}\times [\min\{\pi_{\mathbb{R}}(A)\},\max\{\pi_{\mathbb{R}}(A)\}]$ with $\tilde{V} = \pi_V(A)$ if $\# \pi_V(A) = 1$ and $\tilde{V} = V$ else. Thus, we have to show that for every $A \subseteq G$ with $\#A = \Nbb\setminus\{2,3,4\}$ is not an ufg-premise. For $g \in G$ we have $\gamma(\{g\})= \{g\}$ which is a contradiction to (C1) in Definition~4.2. of the main article. For the upper bound, we first utilize that the formal context can be divided into three formal context $\Kbb_{\text{spatial}}, \Kbb_{elevation}$ and $\Kbb_{vegetation}$. One can easily show that the ufg-premises of these formal contexts are bounded from above by $3, 2$ and $2$. Hence, applying  Lemma~\ref{lem: ufg bound joined fc} provides us with an upper bound of $7$. To show that the maximal cardinality is $4$, we prove for cardinalities 5, 6 and 7 directly that they cannot be an ufg-premise. 
	
	So let $A\subseteq G$ with $\#A = 5$ and $g \in \gamma(A)$. Then there exists $A_1 \subseteq A$ with $\#A_1 = 3$, so that $\pi_{\Rbb^2}(g) \in \pi_{\Rbb^2}\circ\gamma(A_1)$. If $\pi_{\Rbb}(g) \in \pi_{\Rbb}\circ\gamma(A_1) = [\min\{\pi_{\mathbb{R}}(A)\}, \max\{\pi_{\mathbb{R}}(A)\}]$ let $\tilde{g}$ be another point in $A$ with a different vegetation than that in $A_1$ (if it doesn't exist, just take an arbitrary one) and we set $A_g = A_1\cup \tilde{g}$. Then $\# A_g =4$ and $g \in \gamma(A_g)$ is true. If $\pi_{\Rbb}(g) \not\in \pi_{\Rbb}\circ\gamma(A_1)$, then there exists $\overline{g} \in A$ such that $\pi_{\Rbb}(g) \in \pi_{\Rbb}\circ\gamma(A_1 \cup \overline{g})$. Now look at these four elements $A_1 \cup \overline{g}$, because of the geometry in $\Rbb^2$ (i.e. there are only two cases, either $\pi_{\Rbb^2}(\overline{g}) \not in \pi_{\Rbb^2}\circ\gamma(A_1)$ or $\pi_{\Rbb^2}(\overline{g}) \in \pi_{\Rbb^2}\circ\gamma(A_1)$), there exists a subset $A_2 \subsetneq A_1$ with $\pi_{\Rbb^2}(g) \in \pi_{\Rbb^2}\circ \gamma(A_2\cup \overline{g})$. Now we are back in uppercase, and using the uppercase argument, we can define $A_g$ with $\# A_g =4$ so that $g \in \gamma(A_g)$.
	
	This can be done for every $g \in \gamma(A)$ and we obtain a division of $\gamma(A)$ by $\cup_{g \in A} \gamma(A_g) = \gamma(A)$. This is a contradiction to the union-free condition (C2). Hence, $A$ with $\#A = 5$ cannot be an ufg-premise. 
	
	Similar one can show that $\# A = 6,7$ cannot be an ufg-premise either which gives the claim.
\end{proof}

\section{Quasiconcavity from the Perspective of Loss Functions}

In this section, we shortly outline that a quasiconcave version of a depth function $D^{qc}$ can be also defined to be the depth function that has minimal loss w.r.t.~one specific loss function. In general we define:
\begin{definition}
Let $D(\cdot, \Kbb, P)$ be a depth function on $G$ with corresponding formal context $\Kbb$ and probability measure $P$. Let $L$ be a loss function on the function space and $\mathcal{Q}$ a subset of quasiconcave functions on $G$ based on $\Kbb$. We say that a depth function $\tilde{D}(\cdot, \Kbb, P)$ is a \textit{close quasiconcave version of $D$ w.r.t.~$\Kbb, P, L$ and $\mathcal{Q}$} if and only if
\begin{align*}
	D^{qc, L}(\cdot, \Kbb, P) = \arg\min\limits_{E(\cdot, \Kbb, P)\in \mathcal{Q}} \int_G L(E(\cdot, \Kbb, P), D(\cdot, \Kbb, P)) dP.
\end{align*}
\end{definition}
Note that this definition is only well-defined when the minimum is attained and that the $D^{qc, L}(\cdot, \Kbb, P)$ is only unique except for a null set.

Let us now consider the special case of the following loss function $L: \mathbb{R} \times \mathbb{R} \to \Rbb\cup\{\infty\}, (x,y) \mapsto (+ \infty) 1_{x<y} + (x-y)1_{x \ge y}$.
Thus, when looking at the order provided by a depth function, we force that the resulting quasiconcave depth function only reorders upwards and not downwards (except for null sets). 


\begin{theorem}
	Let $D(\cdot, \Kbb, P)$ be a depth based on formal concept analysis on a finite object set $G$ and for every $g \in G$ we have $P(g)>0$. Let $\Qcal$ be the set of all quasiconcave functions on $G$ where the quasiconcavity is defined by $\Kbb$. Then $D^{qc}$ is a close quasiconcave version of $D$ w.r.t.~$\Kbb, P, L$ and $\mathcal{Q}$.
\end{theorem}
\begin{proof}
	The proof that $D^{qc}$ is quasiconcave follows from Theorem~\ref{th:tildeDqc}. We show that $D^{qc}$ is a closed quasiconcave version of $D$ w.r.t.~$\Kbb, P, L$ and $\mathcal{Q}$. Note that reordering any object in $D^{qc}$ lower than in $D$ already gives an infinite loss (since every object in $G$ has a positive probability). So the depth function must be the smallest quasiconcave depth function that has $D$ as a point-wise lower bound. This is exactly $D^{qc}$.
\end{proof}

\bibliographystyle{Chicago}

\bibliography{Bibliography-MM-MC}
\end{document}